\documentclass[abstract]{scrartcl}
\usepackage{algorithm}
\usepackage{algpseudocode}
\usepackage{amsmath}
\usepackage{amssymb}
\usepackage{amsthm}
\usepackage[british]{babel}
\usepackage{booktabs}
\usepackage[inline]{enumitem}
\usepackage{etoolbox}
\usepackage[english, iso]{isodate}
\usepackage[utf8]{inputenc}
\usepackage{microtype}
\usepackage{multirow}
\usepackage{stmaryrd}
\usepackage{thmtools}
\usepackage{tikz}
\usepackage{todonotes}

\usepackage{hyperref}
\usepackage[capitalise]{cleveref}

\usetikzlibrary{arrows, automata, calc, matrix, fit, petri, positioning}

\theoremstyle{definition}
\declaretheorem[qed=$\Box$, name=Definition, numberwithin=section]{definition}
\declaretheorem[qed=$\Box$, name=Example, sibling=definition]{example}
\declaretheorem[qed=$\blacksquare$, name=Observation, sibling=definition]{observation}
\declaretheorem[name=Corollary, refname={Corollary,Corollaries}, sibling=definition]{corollary}
\declaretheorem[name=Theorem, sibling=definition]{theorem}
\declaretheorem[name=Proposition, sibling=definition]{proposition}
\declaretheorem[name=Lemma, sibling=definition]{lemma}

\declaretheorem[name=Proof idea, style=remark, numbered=no]{proofidea}

\crefname{figure}{Figure}{Figures}

\algrenewcommand{\algorithmicrequire}{\textbf{Input:}}
\algrenewcommand{\algorithmicensure}{\textbf{Output:}}

\algrenewcommand\algorithmicindent{1em}
\algrenewcommand\textproc{\textit}

\makeatletter
\AtBeginDocument{%
  \let\c@figure\c@algorithm
  
  \let\c@table\c@algorithm
  
}
\makeatother

\setlist[enumerate]{label=(\roman*)}

\allowdisplaybreaks

\DeclareUnicodeCharacter{00AC}{\neg}
\DeclareUnicodeCharacter{00B1}{\pm}
\DeclareUnicodeCharacter{00D7}{\times}
\DeclareUnicodeCharacter{00F7}{\div}
\DeclareUnicodeCharacter{0391}{A}
\DeclareUnicodeCharacter{0392}{B}
\DeclareUnicodeCharacter{0393}{\varGamma}
\DeclareUnicodeCharacter{0394}{\varDelta}
\DeclareUnicodeCharacter{0395}{E}
\DeclareUnicodeCharacter{0396}{Z}
\DeclareUnicodeCharacter{0397}{H}
\DeclareUnicodeCharacter{0398}{\varTheta}
\DeclareUnicodeCharacter{0399}{I}
\DeclareUnicodeCharacter{039A}{K}
\DeclareUnicodeCharacter{039B}{\varLambda}
\DeclareUnicodeCharacter{039C}{M}
\DeclareUnicodeCharacter{039D}{N}
\DeclareUnicodeCharacter{039E}{\varXi}
\DeclareUnicodeCharacter{039F}{O}
\DeclareUnicodeCharacter{03A0}{\varPi}
\DeclareUnicodeCharacter{03A1}{P}
\DeclareUnicodeCharacter{03A3}{\varSigma}
\DeclareUnicodeCharacter{03A4}{T}
\DeclareUnicodeCharacter{03A5}{\varUpsilon}
\DeclareUnicodeCharacter{03A6}{\varPhi}
\DeclareUnicodeCharacter{03A7}{X}
\DeclareUnicodeCharacter{03A8}{\varPsi}
\DeclareUnicodeCharacter{03A9}{\varOmega}
\DeclareUnicodeCharacter{03B1}{\alpha}
\DeclareUnicodeCharacter{03B2}{\beta}
\DeclareUnicodeCharacter{03B3}{\gamma}
\DeclareUnicodeCharacter{03B4}{\delta}
\DeclareUnicodeCharacter{03B5}{\varepsilon}
\DeclareUnicodeCharacter{03B6}{\zeta}
\DeclareUnicodeCharacter{03B7}{\eta}
\DeclareUnicodeCharacter{03B8}{\theta}
\DeclareUnicodeCharacter{03B9}{\iota}
\DeclareUnicodeCharacter{03BA}{\kappa}
\DeclareUnicodeCharacter{03BB}{\lambda}
\DeclareUnicodeCharacter{03BC}{\mu}
\DeclareUnicodeCharacter{03BD}{\nu}
\DeclareUnicodeCharacter{03BE}{\xi}
\DeclareUnicodeCharacter{03BF}{o}
\DeclareUnicodeCharacter{03C0}{\pi}
\DeclareUnicodeCharacter{03C1}{\rho}
\DeclareUnicodeCharacter{03C3}{\sigma}
\DeclareUnicodeCharacter{03C2}{\varsigma}
\DeclareUnicodeCharacter{03C4}{\tau}
\DeclareUnicodeCharacter{03C5}{\upsilon}
\DeclareUnicodeCharacter{03C6}{\varphi}
\DeclareUnicodeCharacter{03C7}{\chi}
\DeclareUnicodeCharacter{03C8}{\psi}
\DeclareUnicodeCharacter{03C9}{\omega}
\DeclareUnicodeCharacter{03D1}{\vartheta}
\DeclareUnicodeCharacter{03D5}{\phi}
\DeclareUnicodeCharacter{03D6}{\varpi}
\DeclareUnicodeCharacter{03F1}{\varrho}
\DeclareUnicodeCharacter{03F4}{\Theta}
\DeclareUnicodeCharacter{03F5}{\epsilon}
\DeclareUnicodeCharacter{01F4}{\'{G}}
\DeclareUnicodeCharacter{1E30}{\'{K}}
\DeclareUnicodeCharacter{1E3E}{\'{M}}
\DeclareUnicodeCharacter{1E54}{\'{P}}
\DeclareUnicodeCharacter{1E82}{\'{W}}
\DeclareUnicodeCharacter{01F5}{\'{g}}
\DeclareUnicodeCharacter{1E31}{\'{k}}
\DeclareUnicodeCharacter{1E3F}{\'{m}}
\DeclareUnicodeCharacter{1E55}{\'{p}}
\DeclareUnicodeCharacter{1E83}{\'{w}}
\DeclareUnicodeCharacter{2026}{\ldots}
\DeclareUnicodeCharacter{2102}{\mathbb{C}}
\DeclareUnicodeCharacter{210D}{\mathbb{H}}
\DeclareUnicodeCharacter{2115}{\mathbb{N}}
\DeclareUnicodeCharacter{2119}{\mathbb{P}}
\DeclareUnicodeCharacter{211A}{\mathbb{Q}}
\DeclareUnicodeCharacter{211D}{\mathbb{R}}
\DeclareUnicodeCharacter{2124}{\mathbb{Z}}
\DeclareUnicodeCharacter{2190}{\leftarrow}
\DeclareUnicodeCharacter{2192}{\rightarrow}
\DeclareUnicodeCharacter{21A6}{\mapsto}
\DeclareUnicodeCharacter{21D0}{\Leftarrow}
\DeclareUnicodeCharacter{21D2}{\Rightarrow}
\DeclareUnicodeCharacter{21DD}{\rightsquigarrow}
\DeclareUnicodeCharacter{2200}{\forall}
\DeclareUnicodeCharacter{2202}{\partial}
\DeclareUnicodeCharacter{2203}{\exists}
\DeclareUnicodeCharacter{2204}{\nexists}
\DeclareUnicodeCharacter{2205}{\varnothing}
\DeclareUnicodeCharacter{2208}{\in}
\DeclareUnicodeCharacter{2209}{\notin}
\DeclareUnicodeCharacter{220B}{\ni}
\DeclareUnicodeCharacter{220C}{\not\ni}
\DeclareUnicodeCharacter{220F}{\prod}
\DeclareUnicodeCharacter{2210}{\coprod}
\DeclareUnicodeCharacter{2211}{\sum}
\DeclareUnicodeCharacter{2213}{\mp}
\DeclareUnicodeCharacter{2215}{/}
\DeclareUnicodeCharacter{2216}{\setminus}
\DeclareUnicodeCharacter{2217}{\ast}
\DeclareUnicodeCharacter{2218}{\circ}
\DeclareUnicodeCharacter{2219}{\bullet}
\DeclareUnicodeCharacter{221E}{\infty}
\DeclareUnicodeCharacter{2223}{\mid}
\DeclareUnicodeCharacter{2224}{\nmid}
\DeclareUnicodeCharacter{2225}{\parallel}
\DeclareUnicodeCharacter{2226}{\nparallel}
\DeclareUnicodeCharacter{2227}{\wedge}
\DeclareUnicodeCharacter{2228}{\vee}
\DeclareUnicodeCharacter{2229}{\cap}
\DeclareUnicodeCharacter{222A}{\cup}
\DeclareUnicodeCharacter{2260}{\ne}
\DeclareUnicodeCharacter{2261}{\equiv}
\DeclareUnicodeCharacter{2262}{\not\equiv}
\DeclareUnicodeCharacter{2264}{\leq}
\DeclareUnicodeCharacter{2265}{\geq}
\DeclareUnicodeCharacter{2248}{\approx}
\DeclareUnicodeCharacter{2282}{\subset}
\DeclareUnicodeCharacter{2283}{\supset}
\DeclareUnicodeCharacter{2284}{\not\subset}
\DeclareUnicodeCharacter{2285}{\not\supset}
\DeclareUnicodeCharacter{2286}{\subseteq}
\DeclareUnicodeCharacter{2287}{\supseteq}
\DeclareUnicodeCharacter{2288}{\nsubseteq}
\DeclareUnicodeCharacter{2289}{\nsupseteq}
\DeclareUnicodeCharacter{228F}{\sqsubset}
\DeclareUnicodeCharacter{2290}{\sqsupset}
\DeclareUnicodeCharacter{2291}{\sqsubseteq}
\DeclareUnicodeCharacter{2292}{\sqsupseteq}
\DeclareUnicodeCharacter{2293}{\sqcap}
\DeclareUnicodeCharacter{2294}{\sqcup}
\DeclareUnicodeCharacter{2295}{\oplus}
\DeclareUnicodeCharacter{2296}{\ominus}
\DeclareUnicodeCharacter{2297}{\otimes}
\DeclareUnicodeCharacter{2298}{\oslash}
\DeclareUnicodeCharacter{2299}{\odot}
\DeclareUnicodeCharacter{229A}{\circledcirc}
\DeclareUnicodeCharacter{229B}{\circledast}
\DeclareUnicodeCharacter{229D}{\circleddash}
\DeclareUnicodeCharacter{22A2}{\vdash}
\DeclareUnicodeCharacter{22A3}{\dashv}
\DeclareUnicodeCharacter{22A4}{\top}
\DeclareUnicodeCharacter{22A5}{\bot}
\DeclareUnicodeCharacter{22AC}{\nvdash}
\DeclareUnicodeCharacter{22B2}{\vartriangleleft}
\DeclareUnicodeCharacter{22B3}{\vartriangleright}
\DeclareUnicodeCharacter{22C0}{\bigwedge}
\DeclareUnicodeCharacter{22C1}{\bigvee}
\DeclareUnicodeCharacter{22C2}{\bigcap}
\DeclareUnicodeCharacter{22C3}{\bigcup}
\DeclareUnicodeCharacter{22C5}{\cdot}
\DeclareUnicodeCharacter{22C6}{\star}
\DeclareUnicodeCharacter{22EE}{\vdots}
\DeclareUnicodeCharacter{22EF}{\cdots}
\DeclareUnicodeCharacter{22F1}{\ddots}
\DeclareUnicodeCharacter{27E6}{\llbracket}
\DeclareUnicodeCharacter{27E7}{\rrbracket}
\DeclareUnicodeCharacter{27E8}{\langle}
\DeclareUnicodeCharacter{27E9}{\rangle}
\DeclareUnicodeCharacter{27F8}{\Longleftarrow}
\DeclareUnicodeCharacter{27F9}{\Longrightarrow}

\DeclareUnicodeCharacter{2113}{\ell}
\DeclareUnicodeCharacter{2133}{\mathcal{M}}
\DeclareUnicodeCharacter{228A}{\subsetneq}

\newcommand\widebar[1]{\overline{#1}}
\newcommand\wts[1]{\mathrm{weights}_{#1}}
\renewcommand\hom[1]{\mathrm{hom}_{#1}}

\DeclareMathOperator{\dom}{dom}

\DeclareMathOperator{\fanout}{fan-out}
\DeclareMathOperator{\fromBrackets}{fromBrackets}

\DeclareMathOperator{\pos}{pos}
\DeclareMathOperator{\range}{rng}
\DeclareMathOperator{\rank}{rank}
\DeclareMathOperator{\supp}{supp}
\DeclareMathOperator{\sort}{sort}

\DeclareMathOperator{\yield}{yield}
\DeclareMathOperator{\toBrackets}{toBrackets}

\newcommand{\theBar}{\widebar{(⋅)}}
\newcommand{\DYCK}{\operatorname{DYCK}}
\newcommand{\REG}[1][]{\operatorname{REG}\ifstrempty{#1}{}{(#1)}}
\newcommand{\HOM}[1][]{\operatorname{HOM}\ifstrempty{#1}{}{(#1)}}
\makeatletter
\newcommand{\mDYCK}[1][]{\ifstrempty{#1}{}{#1\text{-}}{\operator@font{}mDYCK}}
\newcommand{\cmDYCK}[1][]{\ifstrempty{#1}{}{#1\text{-}}{\operator@font{}mDYCK}_{\text{c}}}
\newcommand{\MCFG}[1][]{\ifstrempty{#1}{}{#1\text{-}}{\operator@font{}MCFG}}
\newcommand{\MCFL}[1][]{\ifstrempty{#1}{}{#1\text{-}}{\operator@font{}MCFL}}
\newcommand{\aHOM}[1][]{α{\operator@font{}HOM}\ifstrempty{#1}{}{(#1)}}
\makeatother

\newcommand{\upto}[1]{\ensuremath{[#1]}}
\newcommand{\sem}[1]{⟦#1⟧}
\newcommand{\lcm}{\mathrm{lcm}}

\newcommand{\pdpop}[1]{\operatorname{pop}(#1)}
\newcommand{\pdpush}[2]{\operatorname{push}(#1,#2)}
\newcommand{\proofpart}[1]{\par\smallskip\noindent\textbf{#1}\quad}

\newcommand{\proofpartns}[1]{\textbf{#1}\quad}

\title{A Chomsky-Schützenberger representation for weighted multiple context-free languages\footnote{This is an extended version of a paper with the same title presented at the 12th International Conference on Finite-State Methods and Natural Language Processing (FSMNLP 2015), \cite{Den15}.}}

\author{Tobias Denkinger \\
  \normalsize Faculty of Computer Science \\
	\normalsize Technische Universität Dresden \\
	\normalsize 01062 Dresden, Germany \\
  \normalsize \href{mailto:tobias.denkinger@tu-dresden.de}{\nolinkurl{tobias.denkinger@tu-dresden.de}}
}

\begin{document}

\maketitle
\begin{abstract}
	We prove a Chomsky-Schützenberger representation theorem for multiple context-free languages weighted over complete commutative strong bimonoids.
\end{abstract}

\tableofcontents
\clearpage

\section{Introduction}

Mildly context-sensitive languages receive much attention in the natural language processing community \cite{Kal10}.
Many classes of mildly context-sensitive languages are subsumed by the multiple context-free languages, e.g.
	the languages of
		head grammars,
		linear context-free rewriting systems \cite{SekMatFujKas91},
		combinatory categorial grammars \cite{VijWeiJos86,WeiJos88},
		linear indexed grammars \cite{Vij87},
		minimalist grammars, \cite{Mic01,Mic01a}, and
		finite-copying lexical functional grammars \cite{SekNakKajAndKas93}.

The Chomsky-Schützenberger (CS) representation for context-free languages \cite[Proposition~2]{ChoSch63} has been generalised to a variety of unweighted and weighted settings, e.g.
	context-free languages weighted with commutative semirings \cite[Theorem~4.5]{SalSoi78},
	tree adjoining languages \cite[Lemma~3.5.2]{Wei88},
	multiple context-free languages \cite[Theorem~3]{YosKajSek10},
	context-free languages weighted with unital valuation monoids \cite[Theorem~2]{DroVog13},
	yields of simple context-free tree languages \cite[Theorem~8.3]{Kan14},
	indexed languages (\cite[Theorems~1 and~2]{DusParSpe79}; \cite[Theorem~4]{FraVou15}; and \cite[Theorem~18]{FraVou16}), and
	automata with storage weighted with unital valuation monoids \cite[Theorem~11]{HerVog15}.

We give a generalisation to the case of multiple context-free languages weighted with a complete commutative strong bimonoid.
\Cref{sec:Dyck_and_mDyck,sec:weighted_CS} contain the main contributions of this paper.
The outline of these sections is:
\begin{itemize}
\item
	In order to obtain a CS representation for multiple context-free languages, Yoshinaka, Kaji, and Seki~\cite{YosKajSek10} introduce multiple Dyck languages.
	We give a more algebraic definition of multiple Dyck languages using congruence relations together with a decision algorithm for membership that is strongly related to these congruence relations (\cref{sec:Dyck_and_mDyck}).
\item
	In \cref{sec:weighted_CS} we provide a CS representation for weighted multiple context-free languages by means of a modular proof that first separates the weights from the given  grammar and then employs the result for the unweighted case (using the same overall idea as in Droste and Vogler~\cite{DroVog13}).
\end{itemize}

Since our proofs do not require distributivity, we can be slightly more general than complete commutative semirings.
The weight algebras considered here are therefore the complete commutative strong bimonoids.


\section{Preliminaries}
\label{sec:prelim}

In this section we recall formalisms used in this paper and fix some notation:
We denote by $ℕ$ the set of natural numbers (including zero).
For every $n ∈ ℕ$ we abbreviate $\{1, …, n\}$ by $\upto{n}$.
Let $A$ be a set.
The \emph{power set of $A$} is denoted by $\mathcal{P}(A)$.
Let $B$ be a finite set.
A \emph{partition of $B$} is a set $\mathfrak{P} ⊆ \mathcal P(B)$ where the elements of $\mathfrak{P}$ are non-empty, pairwise disjoint, and $⋃_{\mathfrak{p} ∈ \mathfrak{P}} \mathfrak{p} = B$.


Let $A$ and $B$ be sets and $A' ⊆ A$.
The \emph{set of functions from $A$ to $B$} is denoted by $A → B$, we still write $f: A → B$ rather then $f ∈ A → B$.
Let $f$ and $g$ be functions.
The \emph{domain and range of $f$} are denoted by $\dom(f)$ and $\range(f)$, respectively.
The \emph{restriction of $f$ to $A'$}, denoted by $f|_{A'}$, is a function from $A'$ to $B$ such that $f|_{A'}(a') = f(a')$ for every $a' ∈ A'$.
We denote the function obtained by applying $g$ after $f$ by $g ∘ f$.
Let $F$ be a set of functions and $B ⊆ ⋂_{f ∈ F} \dom(f)$.
The set $\{f(B) ∣ f ∈ F\} ⊆ \mathcal{P}(\range(f))$ is denoted by $F(B)$.
Let $G$ and $H$ be sets of functions.
The set $\{ h ∘ g ∣ h ∈ H, g ∈ G\}$ of functions is denoted by $H ∘ G$.

Let $A$ be a set and ${≈} ⊆ A × A$ a binary relation on $A$.
We call ${≈}$ an \emph{equivalence relation (on $A$)} if it is reflexive, symmetric, and transitive.
Let $a ∈ A$ and ${≈}$ be an equivalence relation.
The \emph{equivalence class of $a$ in ${≈}$}, denoted by $[a]_≈$, is $\{b ∈ A ∣ a ≈ b\}$.
Let $f: A^k → A$ be a function.
We say that ${≈}$ \emph{respects $f$} if for every $(a_1, b_1), …, (a_k, b_k) ∈ {≈}$ holds $f(a_1, …, a_k) ≈ f(b_1, …, b_k)$.
Now let $\mathcal{A}$ be an algebra with underlying set $A$.
We call ${≈}$ a \emph{congruence relation (on $\mathcal{A}$)} if $≈$ is an equivalence relation and respects every operation of $\mathcal{A}$.


\subsection*{Sorts}

We will use the concept of sorts to formalise restrictions on building terms (or trees), e.g. derivation trees or terms over functions.
One can think of sorts as data types in a programming language: Every concrete value has a sort (type) and every function requires its arguments to be of fixed sorts (types) and returns a value of some fixed sort (type).

Let $S$ be a countable set (of \emph{sorts}) and $s ∈ S$.
An \emph{$S$-sorted set} is a tuple $(B, \sort)$ where $B$ is a set and $\sort$ is a function from $B$ to $S$.
We denote the preimage of $s$ under $\sort$ by $B_s$ and abbreviate $(B, \sort)$ by $B$; $\sort$ will always be clear from the context.
Let $A$ be an $(S^* × S)$-sorted set.
The \emph{set of terms over $A$}, denoted by $T_A$, is the smallest $S$-sorted set $T$ where $ξ = a(ξ_1, …, ξ_k) ∈ T_s$ if there are $s, s_1, …, s_k ∈ S$ such that $a ∈ A_{(s, s_1⋯s_k)}$ and $ξ_i ∈ T_{s_i}$ for every $i ∈ [k]$.
Let $ξ = a(ξ_1, …, ξ_k) ∈ T_A$.
The \emph{set of positions in $ξ$} is defined as $\pos(ξ) = \{ε\} ∪ \{ iu ∣ i ∈ [k], u ∈ \pos(ξ_i) \}$ and for every $π ∈ \pos(ξ)$ the \emph{symbol in $ξ$ at position $π$} is defined as $ξ(π) = a$ if $π = ε$ and as $ξ(π) = ξ_i(u)$ if $π=iu$ for some $i ∈ [k]$ and $u ∈ \pos(ξ_i)$.

\subsection*{Weight algebras}

A \emph{monoid} is an algebra $(\mathcal{A}, {⋅}, 1)$ where~${⋅}$ is associative and~$1$ is neutral with respect to~${⋅}$.
A \emph{bimonoid} is an algebra $(\mathcal{A}, {+}, {⋅}, 0, 1)$ where $(\mathcal{A}, {+}, 0)$ and $(\mathcal{A}, {⋅}, 1)$ are monoids.
We call a bimonoid \emph{strong} if $(\mathcal{A}, {+}, 0)$ is commutative and for every $a ∈ \mathcal{A}$ we have $0 ⋅ a = 0 = a ⋅ 0$.
Intuitively, a strong bimonoid is a semiring without distributivity.
A strong bimonoid is called \emph{commutative} if $(\mathcal{A}, {⋅}, 1)$ is commutative.
A commutative strong bimonoid is \emph{complete} if there is an infinitary sum operation ${∑}$ that
	maps every indexed family of elements of $\mathcal{A}$ to $\mathcal{A}$,
	extends ${+}$, and
	satisfies infinitary associativity and commutativity laws \cite[Section~2]{DroVog13}:
\begin{enumerate}
	\item $∑_{i ∈ ∅} a(i) = 0$;
	\item for every $j ∈ I: ∑_{i ∈ \{j\}} a(i) = a(j)$;
	\item for every $j, k ∈ I$ with $j ≠ k: ∑_{i ∈ \{j, k\}} a(i) = a(j) + a(k)$; and
	\item for every countable set $J$ and family $\bar I: J → \mathcal{P}(I)$ with $I = ⋃_{j ∈ J} \bar I(j)$ and for every $j, j' ∈ J$ with $j ≠ j' \implies \bar I(j) ∩ \bar I(j') = ∅$, we have $∑_{j ∈ J}∑_{i ∈ I(j)} a(i) = ∑_{i ∈ I} a(i)$.
\end{enumerate}
For the rest of this paper let $(\mathcal{A}, {+}, {⋅}, 0, 1)$, abbreviated by $\mathcal{A}$, be a complete commutative strong bimonoid.

\begin{example}\label{ex:bimonoids}
	We provide a list of complete commutative strong bimonoids \cite[Example~1]{DroStueVog10} some of which are relevant for natural language processing:
	\begin{itemize}
		\item Any complete commutative semiring, e.g.
		\begin{itemize}
			\item the \emph{Boolean semiring} \( \mathbb B = \big( \{0, 1\}, {∨}, {∧}, 0, 1 \big) \),
			\item the \emph{probability semiring} \( \operatorname{Pr} = \big( ℝ_{≥ 0}, {+}, {⋅}, 0, 1\big) \),
			\item the \emph{Viterbi semiring} \( \big( [0,1], {\max}, {⋅}, 0, 1 \big) \),
			\item the \emph{tropical semiring} \( \big( ℝ ∪ \{∞\}, {\min}, {+}, ∞, 0 \big) \),
			\item the \emph{arctic semiring} \( \big( ℝ ∪ \{-∞\}, {\max}, {+}, -∞, 0 \big) \),
		\end{itemize}
		\item any complete lattice, e.g.
		\begin{itemize}
			\item any non-empty finite lattice $\big( L, {∨}, {∧}, 0, 1 \big)$ where $L$ is a non-empty finite set,
			\item the lattice $(\mathcal{P}(A), {∪}, {∩}, ∅, A)$ where $A$ is an arbitrary set,
			\item the lattice $(ℕ, \lcm, \gcd, 1, 0)$,
		\end{itemize}
		\item the \emph{tropical bimonoid} \(\big( ℝ_{≥0} ∪ \{∞\}, {+}, {\min}, 0, ∞ \big)\),
		\item the \emph{arctic bimonoid} \(\big( ℝ_{≥0} ∪ \{-∞\}, {+}, {\max}, 0, -∞ \big)\), and
		\item the algebras \( \operatorname{Pr}_1 = ([0,1], {⊕_1}, {⋅}, 0, 1) \) and \( \operatorname{Pr}_2 = ([0,1], {⊕_2}, {⋅}, 0, 1) \) where $a ⊕_1 b = a + b - a ⋅ b$ and $a ⊕_2 b = \min \{a + b, 1\}$ for every $a, b ∈ [0,1]$.
	\end{itemize}
	where $ℝ$ and $ℝ_{≥ 0}$ denote the set of reals and the set of non-negative reals, respectively; ${+}$, ${⋅}$, ${\max}$, ${\min}$ denote the usual operations; ${∧}$, ${∨}$ denote disjunction and conjunction, respectively, for the boolean semiring and join and meet, respectively, for any non-empty finite lattice; and $\lcm$ and $\gcd$ are binary functions that calculate the least common multiple and the greatest common divisor, respectively.
	
	Also, there are some bimonoids that are interesting for natural language processing but are \emph{not} complete commutative strong bimonoids.
	E.g.
	\begin{itemize}
		\item the \emph{semiring of formal languages} \( \big( \mathcal{P}(Σ^*), {∪}, {⋅}, ∅, \{ε\} \big) \) where $Σ$ is an alphabet and ${⋅}$ is language concatenation, i.e. $L_1 ⋅ L_2 = \{uv ∣ u ∈ L_1, v ∈ L_2 \}$ for every $L_1, L_2 ⊆ Σ^*$; and
		\item the semiring \( \big( Σ^* ∪ \{∞\}, {∧}, {⋅}, ∞, ε \big) \) where $Σ$ is an alphabet, ${⋅}$ is concatenation, ${∧}$ calculates the longest common prefix of its arguments, and $∞$ is a new element that is neutral with respect to ${∧}$ and annihilating with respect to ${⋅}$ \cite{Moh00}.
	\end{itemize}
	Both examples are not commutative.
\end{example}

An \emph{$\mathcal{A}$-weighted language (over $Δ$)} is a function $L: Δ^* → \mathcal{A}$.
The \emph{support of $L$}, denoted by $\supp(L)$, is $\{w ∈ Δ^* ∣ L(w) ≠ 0\}$.
If $\lvert\supp(L)\rvert ≤ 1$, we call $L$ a \emph{monomial}.
We write $μ.w$ for $L$ if $L(w) = μ$ and for every $w' ∈ Δ^*∖ \{w\}$ we have $L(w') = 0$.

\subsection*{Recognisable languages}

For the many known results concerning finite state automata and regular languages, we will rely on \cite{HopUll69} and \cite{HopUll79}.
We nevertheless recall the basic definitions:

\begin{definition}
	A \emph{finite state automaton}, for short: FSA, is a tuple $ℳ = (Q, Δ, q_0, F, T)$ where
		$Δ$ is an alphabet (\emph{terminals}),
		$Q$ is a finite set (\emph{states}) disjoint from $Δ$,
		$q_0 ∈ Q$ (\emph{initial state}),
		$F ⊆ Q$ (\emph{final states}), and
		$T ⊆ Q × Δ^* × Q$ is a finite set (\emph{transitions}).
\end{definition}

A \emph{run in $ℳ$} is a string $κ ∈ (Q ∪ Δ)^*$ where
	for every substring of $κ$ of the form $quq'$ (for some $q, q' ∈ Q$ and $u ∈ Δ^*$) we have that
		$(q, u, q') ∈ T$, we say that the transition $(q, u, q')$ occurs in $κ$,
		the first symbol of $κ$ is $q_0$, and
		the last symbol of $κ$ is in $F$.
The \emph{word corresponding to $κ$} is obtained by removing the elements of $Q$ from $κ$.
The \emph{language of $ℳ$} is denoted by $L(ℳ)$.
The \emph{set of recognisable languages}, denoted by $\REG$, is the set of languages $L$ for which there is an FSA $ℳ$ with $L = L(ℳ)$.

%

\subsection*{Weighted string homomorphisms}

\begin{definition}
  Let $Δ$ and $Γ$ be alphabets and $g: Δ → (Γ^* → \mathcal{A})$ (i.e. a function that takes an element of $Δ$ and returns a function that takes an element of $Γ^*$ and returns an element of $\mathcal{A}$) such that $g(δ)$ is a monomial for every $δ ∈ Δ$.
  We define $\widehat{g}: Δ^* → (Γ^* → \mathcal{A})$ where for every $k ∈ ℕ$, $δ_1, …, δ_k ∈ Δ$, and $u ∈ Γ^*$ we have
  \[\widehat{g}(δ_1 ⋯ δ_k)(u) = ∑_{\substack{u_1, …, u_k ∈ Γ^* \\ u = u_1 ⋯
        u_k}} g\big(δ_1\big)\big(u_1\big) ⋅ … ⋅ g\big(δ_k\big)\big(u_k\big)\,.\]
  We call $\widehat{g}$ an \emph{$\mathcal{A}$-weighted (string) homomorphism}.
\end{definition}

Clearly, $\widehat{g}(u)$ is a monomial for every $u ∈ Δ^*$.
We call $\widehat{g}$ \emph{alphabetic} if there is a function $h: Δ → (Γ ∪ \{ε\} → \mathcal{A})$ with $\widehat{g} = \widehat{h}$.
If $\widehat{g}(u) = μ.w$ for $u ∈ Δ^*$, then we will sometimes say “$\widehat{g}$ maps $u$ to $w$” (leaving out the weight $μ$) or “$\widehat{g}$ weights $u$ with $μ$” (leaving out the word $w$).
Now assume that $\mathcal{A} = \mathbb{B}$ and we have $\lvert \supp(g(δ)) \rvert = 1$ for every $δ ∈ Δ$.
Then $g$ can be construed as a function from $Δ$ to $Γ^*$ and $\widehat{g}$ can be construed as a function from $Δ^*$ to $Γ^*$.
In this case we call $\widehat{g}$ a \emph{(string) homomorphism}.
%
The sets of all $\mathcal{A}$-weighted homomorphisms, $\mathcal{A}$-weighted alphabetic homomorphisms, homomorphisms, and alphabetic homomorphisms are denoted by $\HOM[\mathcal{A}]$, $\aHOM[\mathcal{A}]$, $\HOM$, and $\aHOM$, respectively.

\subsection*{Weighted multiple context-free languages}

We fix a set $X = \{ x_i^j ∣ i, j ∈ ℕ_+ \}$ of \emph{variables}.
Let $Δ$ be an alphabet.
The set of \emph{composition representations over $Δ$} is the $(ℕ^*×ℕ)$-sorted set $\mathrm{RF}_Δ$ where for every $s_1, …, s_{\ell}, s ∈ ℕ$ we define
	$X_{(s_1 ⋯ s_{\ell}, s)} = \{x_i^j ∣ i ∈ [\ell], j ∈ [s_i] \} ⊆ X$ and
	$(\mathrm{RF}_Σ)_{(s_1 ⋯ s_{\ell}, s)}$ as the set that contains $[u_1, …, u_s]_{(s_1 ⋯ s_{\ell}, s)}$ for every $u_1, …, u_s ∈ (Δ ∪ X_{(s_1 ⋯ s_{\ell}, s)})^*$.
We will often write $X_f$ instead of $X_{(s_1 ⋯ s_{\ell}, s)}$.
Let $f = [u_1, …, u_s]_{(s_1 ⋯ s_{\ell}, s)} ∈ \mathrm{RF}_Σ$.
The \emph{string function of $f$}, also denoted by $f$, is the function from $(Δ^*)^{s_1} × ⋯ × (Δ^*)^{s_{\ell}}$ to $(Δ^*)^s$ such that
\(
	f((w_1^1, …, w_1^{s_1}), …, (w_{\ell}^1, …, w_{\ell}^{s_{\ell}}))
		= (u_1', …, u_s')
\)
where $(u_1', …, u_s')$ is obtained from $(u_1, …, u_s)$ by replacing each occurrence of $x_i^j$ by $w_i^j$ for every $i ∈ [\ell]$ and $j ∈ [s_{\ell}]$.
The set of all string functions for some composition representation over $Δ$ is denoted by $\mathrm{F}_Δ$.
From here on we no longer distinguish between composition representations and string functions.
We define the \emph{rank of $f$}, denoted by $\rank(f)$, and the \emph{fan-out of $f$}, denoted by $\fanout(f)$, as $\ell$ and $s$, respectively.
The string function
	$f$ is called \emph{linear} if in $u_1 ⋯ u_s$ every element of $X_f$ occurs at most once,
	$f$ is called \emph{non-deleting} if in $u_1 ⋯ u_s$ every element of $X_f$ occurs at least once, and
	$f$ is called \emph{terminal-free} if $u_1, …, u_s ∈ X_f^*$.
The subscript is dropped from the string function if its sort is clear from the context.

Note that for every $s' ∈ ℕ^* × ℕ$, the set of linear terminal-free string functions of sort $s'$ is finite.

\begin{definition}
	A \emph{multiple context-free grammar (over $Δ$)}, for short: {($Δ$-)MCFG}, is a tuple $(N, Δ, S, P)$ where
		$N$ is a finite $ℕ$-sorted set (\emph{non-terminals}),
		$S ∈ N_1$ (\emph{initial non-terminal}), and
    $P$ is a finite set (\emph{productions}) such that
    \begin{align*}
      P ⊆_{\text{fin}} \big\{ (A, f, A_1 ⋯ A_ℓ) ∈ N × F_Δ × N^ℓ &∣ \sort(f) = (\sort(A_1)⋯ \sort(A_ℓ), \sort(A)), \\*
                                                                &\quad f \text{ is linear}, ℓ ∈ ℕ \big\}\text{.}
    \end{align*}
	We construe $P$ as an $(N^*×N)$-sorted set where for every $ρ = (A, f, A_1 ⋯ A_ℓ) ∈ P$ we have $\sort(ρ) = (A, A_1 ⋯ A_ℓ)$.
\end{definition}

Let $G = (N, Δ, S, P)$ be an MCFG and $w ∈ Δ^*$.
A production $(A, f, A_1 ⋯ A_ℓ) ∈ P$ is usually written as $A → f(A_1, …, A_ℓ)$; it inherits $\rank$ and $\fanout$ from $f$.
Also, $\rank(G) = \max_{ρ ∈ P} \rank(ρ)$ and $\fanout(G) = \max_{ρ ∈ P} \fanout(ρ)$.
MCFGs of fan-out at most $k$ are called \emph{$k$-MCFGs}.

The function $\yield: T_P → (Σ^*)^*$ assigns to every tree $d ∈ T_P$ the string obtained by projecting every production in $d$ to the contained function (i.e. the second component) and then interpreting the resulting term over $F_Δ$.

Let $A ∈ N$.
The \emph{set of subderivations in $G$ from $A$}, denoted by $D_G(A)$, is the set of all terms over $P$ with sort $A$, i.e. $D_G(A) = (T_P)_A$.
The \emph{set of derivations in $G$} is $D_G = D_G(S)$.
Let $w ∈ Σ^*$.
The \emph{set of derivations of $w$ in $G$} is $D_G(w) = \{ d ∈ D_G ∣ \yield(d) = (w)\}$.

The \emph{language of $G$} is $L(G) = \{ w ∈ Δ^* ∣ D_G(w) ≠ ∅ \}$.
A language $L$ is called \emph{multiple context-free} if there is an MCFG $G$ with $L = L(G)$.
The \emph{set of multiple context-free languages (for which a $k$-MCFG exists)} is denoted by $\MCFL$ ($\MCFL[k]$, respectively).

The language class $\MCFL[k]$ is a substitution-closed full abstract family of languages \cite[Theorem~3.9]{SekMatFujKas91}.
In particular, $\MCFL[k]$ is closed under intersection with regular languages and under homomorphisms.

\begin{definition}
	An \emph{$\mathcal{A}$-weighted MCFG (over $Δ$)} is a tuple
		$(N, Δ, S, P, μ)$
	such that
		$(N, Δ, S, P)$ is an MCFG and
		$μ: P → \mathcal{A} ∖ \{0\}$ (\emph{weight assignment}).
\end{definition}

Let $G = (N, Δ, S, P, μ)$ be an $\mathcal{A}$-weighted MCFG and $w ∈ Δ^*$.
The \emph{set of derivations of $w$ in $G$} is the set of derivations of $w$ in $(N, Δ, S, P)$.
$G$ inherits $\fanout$ from $(N, Δ, S, P)$;
$\mathcal{A}$-weighted MCFGs of fan-out at most $k$ are called \emph{$\mathcal{A}$-weighted $k$-MCFGs}.
We define a function $\widehat{μ}: D_G → \mathcal{A}$ that applies $μ$ at every position of a given derivation and then multiplies the resulting values (in any order, since ${⋅}$ is commutative).

The \emph{$\mathcal{A}$-weighted language induced by $G$} is the function $\sem{G}: Δ^* → \mathcal{A}$ where for every $w ∈ Δ^*$ we have $\sem{G}(w) = ∑_{d ∈ D_G(w)} \widehat{μ}(d)$.
Two ($\mathcal{A}$-weighted) MCFGs are \emph{equivalent} if they induce the same ($\mathcal{A}$-weighted) language.
An $\mathcal{A}$-weighted language $L$ is called \emph{multiple context-free (of fan-out $k$)} if there is an $\mathcal{A}$-weighted $k$-MCFG $G$ such that $L = \sem{G}$;
$\MCFL[k](\mathcal{A})$ denotes the \emph{set of multiple context-free $\mathcal{A}$-weighted languages of fan-out $k$}.

\begin{example}\label{ex:wmcfg}
	Consider the $\mathrm{Pr}_2$-weighted MCFG $G = \big( N, Δ, S, P, μ  \big)$
	where
		$N_1 = \{S\}$,
		$N_2 = \{A, B\}$,
		$N = N_1 ∪ N_2$,
		$Δ = \{a, b, c, d\}$, and
		$P$ and $μ$ are given by
	\begin{align*}
		P: \enspace
		ρ_1 &= S → [x_1^1x_2^1x_1^2x_2^2](A, B)
		& μ: \enspace
		μ(ρ_1) &= 1 \\*
		ρ_2 &= A → [ax_1^1, cx_1^2](A)
		& μ(ρ_2) &= 1/2 \\
		ρ_3 &= B → [bx_1^1,dx_1^2](B)
		& μ(ρ_3) &= 1/3 \\
		ρ_4 &= A → [ε,ε]()
		& μ(ρ_4) &= 1/2 \\*
		ρ_5 &= B → [ε,ε]()
		& μ(ρ_5) &= 2/3\;.
	\end{align*}
	We observe that $\supp(\sem{G}) = \{a^mb^nc^md^n ∣ m, n ∈ ℕ\}$ and for every $m, n ∈ ℕ$ we have
	\begin{align*}
		\sem{G}(a^mb^nc^md^n)
		&= μ(ρ_1) ⋅ \big(μ(ρ_2)\big)^m ⋅ μ(ρ_4) ⋅ \big(μ(ρ_3)\big)^m ⋅ μ(ρ_5) \\*
		&= 1/(2^m ⋅ 3^{n+1})\text{.}
	\end{align*}
	The only derivation of $w = ac$ in $G$ is shown in \cref{fig:deriv}, its weight and hence also the weight of $w$ is $1/(2^1 ⋅ 3^{0+1}) = 1/6$.
\end{example}

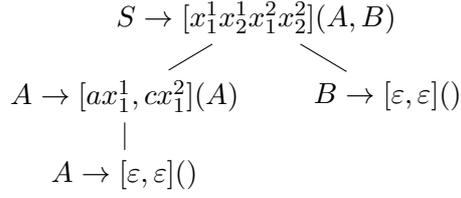
\begin{figure}
  \centering
	\begin{tikzpicture}[level distance=2.7em]
	\node {$S → [x_1^1x_2^1x_1^2x_2^2](A, B)$}
		child [sibling distance=9em] {node {$A → [ax_1^1, cx_1^2](A)$}
				child {node {$A → [ε,ε]()$}}
		}
			child [sibling distance=9em] {node {$B → [ε,ε]()$}};
\end{tikzpicture}
	\caption{Only derivation of $ac$ in $G$ (\cref{ex:wmcfg})}
	\label{fig:deriv}
\end{figure}

A non-terminal is called \emph{productive in an ($\mathcal{A}$-weighted) MCFG} if there is at least one subderivation starting from this non-terminal.
It is obvious that every ($\mathcal{A}$-weighted) $k$-MCFL can be recognised by an ($\mathcal{A}$-weighted) $k$-MCFG that only has productive non-terminals.

\subsubsection*{Non-deleting normal form}

An ($\mathcal A$-weighted) MCFG is called \emph{non-deleting} if the string function in every production is linear and non-deleting.
Seki, Matsumura, Fujii, and Kasami~\cite[Lemma~2.2]{SekMatFujKas91} proved that for every $k$-MCFG there is an equivalent non-deleting $k$-MCFG.
We generalise this to $\mathcal A$-weighted MCFGs.

\begin{lemma}
\label{lem:normalform}
	For every $\mathcal{A}$-weighted $k$-MCFG there is an equivalent non-deleting $\mathcal{A}$-weighted $k$-MCFG.
\end{lemma}
\begin{proofidea}
	We modify the construction for the unweighted case \cite[Lemma~2.2]{SekMatFujKas91} such that it preserves the structure of derivations.
	Then a weight assignment can be defined in an obvious manner.
\end{proofidea}
\begin{proof}
	Let $G = (N, Δ, S, P, μ)$ be an $\mathcal{A}$-weighted $k$-MCFG.
	When examining the proof of Seki, Matsumura, Fujii, and Kasami~\cite[Lemma~2.2]{SekMatFujKas91}, we notice that only step~2 of Procedure~1 deals with non-deletion.
	We construct $N'$ and $P'$ from $(N, Δ, S, P)$ by step~2 of Procedure~1, but drop the restriction that $Ψ ≠ \upto{\sort(A)}$.\footnote{This construction may therefore create productions of fan-out 0.}
	
	Let $g: P' → P$ assign to every $ρ' ∈ P'$ the production in $G$ it has been constructed from.
	Furthermore, let $\widehat{g}: D_{G'} → D_G$ be the function obtained by applying $g$ point-wise.
	We show the following hypothesis by induction on the structure of subderivations:
	
	\noindent
	\emph{Induction hypothesis:} For every $A ∈ N$ and $Ψ ∈ M(A): \widehat{g}$ is a bijection between $D_{G'}(A[Ψ])$ and $D_G(A)$.
	
	\noindent
	\emph{Induction step:} Let $d ∈ D_G(A)$ and $Ψ ∈ M(A)$ with $d = ρ(d_1, …, d_k)$ for some production $ρ ∈ P$ and derivations $d_1 ∈ D_G(A_1)$,~$…$, $d_k ∈ D_G(A_k)$.
	The construction defines $Ψ_1 ∈ M(A_1)$,~$…$, $Ψ_k ∈ M(A_k)$ and a production $ρ'$ which is unique for every $ρ$ and $Ψ$.
	By the induction hypothesis, we know that there are derivations $d_1', …, d_k'$ which are unique for $(d_1, Ψ_1), …, (d_k, Ψ_k)$, respectively.
	Therefore, $d' = ρ(d_1', …, d_k')$ is unique for $d$ and $Ψ$.
	Hence for every $Ψ$, $\widehat{g}$ induces a bijection between $D_{G'}(A[Ψ])$ and $D_G(A)$.

	By construction, the new start symbol is $S[∅]$; hence for the elements of $D_{G'}$, we set $Ψ = ∅$ and by induction hypothesis we obtain that $\widehat{g}$ is bijective.
	Since $\widehat{g}$ preserves the structure of derivations and is a bijection we obtain $\widehat{μ ∘ g} = \widehat{μ} ∘ \widehat{g}$.
	Hence $\sem{(N', Δ, S[∅], P', μ ∘ g)} = \sem{G}$.
	The fan-out is not increased by this construction.
\end{proof}


\section{Multiple Dyck languages}
\label{sec:Dyck_and_mDyck}

According to Kanazawa~\cite[Section~1]{Kan14} there is no definition of multiple Dyck languages using congruence relations.
We close this gap by giving such a definition (\cref{def:mDYCK}).

\subsection{The original definition}\label{sec:definition-mDYCK}

We recall the definition of multiple Dyck languages \cite[Definition~1]{YosKajSek10}:
\begin{definition}
Let
	$Δ$ be a finite $ℕ$-sorted set,\footnote{In Yoshinaka, Kaji, and Seki~\cite{YosKajSek10}, $ℕ$-sorted sets are called indexed sets and $\sort$ is denoted as $\dim$.}
	$\theBar$ be a bijection between $Δ$ and some alphabet $\widebar{Δ}$,
	$k = \max_{δ ∈ Δ} \sort(δ)$, and
	$r ≥ k$.
The \emph{multiple Dyck grammar with respect to $Δ$} is the $k$-MCFG
	$G_Δ = \big(\{A_1, …, A_k\}, \widehat{Δ}, A_1, P \big)$ where
	$\widehat{Δ} = \{ δ^{[i]}, \bar{δ}^{[i]} ∣ δ ∈ Δ, i ∈ \upto{\sort(δ)}\}$,
	$\sort(A_i) = i$ for every $i ∈ \upto{k}$, and
	$P$ is the smallest set such that
\begin{enumerate}
	\item
		for every
			linear non-deleting\footnote{We add the restriction “non-deleting” in comparison to the original definition since the proof of Lemma~1 in Yoshinaka, Kaji, and Seki~\cite{YosKajSek10} only uses non-deleting rules.} terminal-free string function $f ∈ (F_Δ)_{(s_1\cdots s_ℓ, s)}$
		with
			$ℓ ∈ \upto{r}$ and
			$s_1, …, s_ℓ, s ∈ \upto{k}$
		we have \[A_s → f(A_{s_1}, …, A_{s_ℓ}) ∈ P\;\text{,}\]
	\item\label{def:Dyck-MCFG:2}
		for every
			$δ ∈ Δ$ with sort $s$ we have
			\[ A_s → \big[δ^{[1]} x_1^1 \bar{δ}^{[1]}, …, δ^{[s]} x_1^s \bar{δ}^{[s]}\big](A_s) ∈ P\;\text{, and}\]
	\item\label{def:Dyck-MCFG:3}
		for every
			$s ∈ \upto{k}$ we have
		\[A_s → [u_1, …, u_s](A_s) ∈ P\] where
			$u_i ∈ \big\{x_i,\; x_iδ^{[1]}\bar{δ}^{[1]},\; δ^{[1]}\bar{δ}^{[1]}x_i ∣ δ ∈ Δ_1\big\}$ for every $i ∈ \upto{s}$.
\end{enumerate}
The \emph{multiple Dyck language with respect to $Δ$}, denoted by $\mathit{mD}(Δ)$, is $L(G_Δ)$.
We call $\max_{δ ∈ Δ} \sort(δ)$ the \emph{dimension of $\mathit{mD}(Δ)$}.
The \emph{set of multiple Dyck languages of dimension at most $k$} is denoted by $\mDYCK[k]$.
\end{definition}

\subsection{Congruence multiple Dyck languages}

For the rest of this section let $Σ$ be an alphabet.
Also let $\widebar{Σ}$ be a set (disjoint from $Σ$) and $\theBar$ be a bijection between $Σ$ and $\widebar{Σ}$.
Intuitively $Σ$ and $\widebar{Σ}$ are sets of opening and closing parentheses and $\theBar$ matches an opening to its closing parenthesis.

We define $≡_Σ$ as the smallest congruence relation on the free monoid $(Σ ∪ \widebar{Σ})^*$ where for every $σ ∈ Σ$ the cancellation rule $σ \widebar{σ} ≡_{Σ} ε$ holds.
The \emph{Dyck language with respect to $Σ$}, denoted by $\mathit{D}(Σ)$, is $[ε]_{≡_Σ}$.
The \emph{set of Dyck languages} is denoted by $\DYCK$.

\begin{example}\label{ex:Dyck}
	Let
		$Σ = \{(, ⟨, [, ⟦\}$.
		We abbreviate $\bar{(}$, $\bar{⟨}$, $\bar{[}$, and $\bar{⟦}$ by $)$, $⟩$, $]$, and $⟧$, respectively.
	Then we have for example
	\(
		⟦()⟧⟨⟩()
		≡_Σ ⟦⟧⟨⟩
		≡_Σ ⟦⟧
		≡_Σ ε\) and \(
		(⟦)⟧⟨⟩()
		≡_Σ (⟦)⟧()
		≡_Σ (⟦)⟧
		≢_Σ ε
	\).
\end{example}

Let $\mathfrak{P}$ be a partition of $Σ$.
We define $≡_{Σ,\mathfrak{P}}$ as the smallest congruence relation on the free monoid $(Σ ∪ \widebar{Σ})^*$ such that
	if $v_1 \cdots v_ℓ ≡_{Σ, \mathfrak{P}} ε$ with $v_1, …, v_ℓ ∈ \mathit{D}(Σ)$, then the \emph{cancellation rule}
	\begin{align*}
		u_0 σ_1 v_1 \widebar{σ_1} u_1 \cdots σ_ℓ v_ℓ \widebar{σ_ℓ} u_ℓ
		≡_{Σ, \mathfrak{P}} u_0 \cdots u_ℓ
	\end{align*}
	holds for every $\{σ_1, …, σ_ℓ\} ∈ \mathfrak{P}$ and $u_0,\allowbreak …,\allowbreak u_ℓ ∈ \mathit{D}(Σ)$.
Intuitively, every element of $\mathfrak{P}$ denotes a set of \emph{linked} opening parentheses, i.e. parentheses that must be consumed simultaneously by $≡_{Σ, \mathfrak{P}}$.

\begin{definition}\label{def:mDYCK}
	The \emph{congruence multiple Dyck language with respect to $Σ$ and $\mathfrak{P}$}, denoted by $\mathit{mD}_{\text{c}}(Σ, \mathfrak{P})$, is $[ε]_{≡_{Σ, \mathfrak{P}}}$.
\end{definition}

\begin{example}\label{ex:multipleDyck}
	Let
		$Σ = \{(, ⟨, [, ⟦\}$ and
		$\mathfrak{P} = \{ \mathfrak{p}_1, \mathfrak{p}_2 \}$ where
			$\mathfrak{p}_1 = \{ {(}, {⟨} \}$ and
			$\mathfrak{p}_2 = \{ {[}, {⟦} \}$.
		We abbreviate $\bar{(}$, $\bar{⟨}$, $\bar{[}$, and $\bar{⟦}$ by $)$, $⟩$, $]$, and $⟧$, respectively.
	Then we have for example
		\( ⟦()⟧[⟨⟩] ≡_{Σ, \mathfrak{P}} ε \) since
			$\mathfrak{p}_2 = \{ {[}, {⟦} \} ∈ \mathfrak{P}$,
			$()⟨⟩ ≡_{Σ, \mathfrak{P}} ε$, and
			$u_0 = u_1 = u_2 = ε$.
	But \( ⟦()⟧⟨[]⟩ ≢_{Σ, \mathfrak{P}} ε \) since when instantiating the cancellation rule with any of the two elements of $\mathfrak{P}$, we can not reduce $⟦()⟧⟨[]⟩$:
	\begin{enumerate}
		\item
			If we choose $\{ σ_1, σ_2 \} = \{{⟦}, {[}\}$ then we would need to set $u_1 = {⟨}$ and $u_2 = {⟩}$, but they are not in $\mathit{D}(Σ)$, also $() ≢_{Σ, \mathfrak{P}} ε$;
		\item
			If we choose $\{ σ_1, σ_2 \} = \{{(}, {⟨}\}$ then we would need to set $u_0 = {⟦}$ and $u_1 = {⟧}$, but they are not in $\mathit{D}(Σ)$, also $[] ≢_{Σ, \mathfrak{P}} ε$.
	\end{enumerate}
	Hence $⟦()⟧[⟨⟩], ()⟨⟩ ∈ \mathit{mD}_{\text{c}}(Σ, \mathfrak{P})$ and $⟦()⟧⟨[]⟩ ∉ \mathit{mD}_{\text{c}}(Σ, \mathfrak{P})$.
\end{example}

\begin{observation}\label{obs:mDyck_permutation_rule}
	From the definition of $≡_{Σ, \mathfrak{P}}$ it is easy to see that for every $u_1, …, u_k ∈ \mathit{D}(Σ)$ and $v_1, …, v_ℓ ∈ \mathit{D}(Σ)$ we have that $u_1 \cdots u_k,\allowbreak v_1 \cdots v_ℓ ∈ \mathit{mD}_{\text{c}}(Σ, \mathfrak{P})$ implies that every permutation of $u_1, …, u_k,\allowbreak v_1, …, v_ℓ$ is in $\mathit{mD}_{\text{c}}(Σ, \mathfrak{P})$.
\end{observation}

The \emph{dimension of $\mathit{mD}_{\text{c}}(Σ, \mathfrak{P})$} is $\max_{\mathfrak{p} ∈ \mathfrak{P}} \lvert \mathfrak{p} \rvert$.
The \emph{set of congruence multiple Dyck languages (of at most dimension $k$)} is denoted by $\cmDYCK$ ($\cmDYCK[k]$, respectively).

\begin{proposition}
\label{prop:congruence_mDyck}
	\enspace$\mDYCK[k] ⊆ \cmDYCK[k]$
\end{proposition}
\begin{proof}
	We show that a tuple $(w_1, …, w_m)$ can be generated in $G_Δ$ from non-terminal $A_m$ if and only if $w_1, …, w_m$ are all Dyck words and $w_1⋯w_m$ is a multiple Dyck word.
	The “only if” we prove by induction on the structure of derivations in $G_Δ$.
	For “if” we construct derivations in $G_Δ$ by induction on the number of applications of the cancellation rule (including the number of applications to reduce the word $v_1 \cdots v_ℓ$ from the definition on the cancellation rule to~$ε$).

	Let $\mathit{mD} ∈ \mDYCK[k]$.
	Then there is an $ℕ$-sorted set $Δ$ such that $\mathit{mD} = \mathit{mD}(Δ)$ and $k ≥ \max_{δ ∈ Δ} \sort(δ)$.
	We define
		$\mathfrak{p}_δ = \{δ^{[i]} ∣ i ∈ \upto{\sort(δ)}\}$ for every $δ ∈ Δ$,
		$Σ = ⋃_{δ ∈ Δ} \mathfrak{p}_δ$, and
		$\mathfrak{P} = \{ \mathfrak{p}_δ ∣ δ ∈ Δ\}$.
	Clearly $\max_{\mathfrak{p} ∈ \mathfrak{P}} \lvert \mathfrak{p} \rvert ≤ k$.
	Thus $\mathit{mD}_{\text{c}}(Σ, \mathfrak{P}) ∈ \mDYCK[k]$.
	Let $\operatorname{Tup}(G_Δ, A)$ denote the set of tuples generated in $G_Δ$ when starting with non-terminal $A$ where $A$ is not necessarily initial.
	In the following we show that
	for every $m ∈ [\max_{δ ∈ Δ} \sort(δ)]$ and $w_1, …, w_m ∈ (Σ ∪ \bar{Σ})^*: $
	\begin{align*}
	(w_1, …, w_m) ∈ \operatorname{Tup}(G_Δ, A_m)
	\iff w_1 \cdots w_m ∈ \mathit{mD}_{\text{c}}(Σ, \mathfrak{P}) ∧ w_1, …, w_m ∈ \mathit{D}(Σ) \tag{*}\label{eq:theIH}
	\end{align*}
	\proofpart{($⇒$)}
	It follows from the definitions of $\operatorname{Tup}$ and $G_Δ$ that
		$(w_1, …, w_m) ∈ \operatorname{Tup}(G_Δ, A_m)$
	implies that
		there are
			a rule $A_m → f(A_{m_1}, …, A_{m_ℓ})$ in $G_Δ$ and
			a tuple $\vec{u}_i = (u_i^1, …, u_i^{m_i})$ in $\operatorname{Tup}(G_Δ, A_{m_i})$ for every $i ∈ [ℓ]$
		such that $f(\vec{u}_1, …, \vec{u}_ℓ) = (w_1, …, w_m)$.
	By applying the induction hypothesis $ℓ$ times, we also have that
		$u_1^1, …, u_1^{m_1}, …, u_ℓ^1, …, u_ℓ^{m_ℓ} ∈ \mathit{D}(Σ)$ and
		$u_1^1 \cdots u_1^{m_1}, …, u_ℓ^1 \cdots u_ℓ^{m_ℓ} ∈ \mathit{mD}(Σ, \mathfrak{P})$.
	We distinguish three cases (each corresponding to one type of rule in $G_Δ$):
	\begin{enumerate}
		\item
			$f$ is linear, non-deleting, and terminal-free.
				Then we have for every $i ∈ [m]$ that $w_i ∈ \{u_1^1, …, u_1^{m_1}, …, u_ℓ^1, …, u_ℓ^{m_ℓ}\}^*$ and therefore also $w_i ∈ \mathit{D}(Σ)$.
				Furthermore, by applying \cref{obs:mDyck_permutation_rule} $(ℓ-1)$ times, we have that $w_1 \cdots w_m ∈ \mathit{mD}_{\text{c}}(Σ, \mathfrak{P})$.
		\item
			$f = [δ^{[1]}x_1^1\bar{δ}^{[1]}, …, δ^{[m]}x_1^m\bar{δ}^{[m]}]$; then
				$ℓ = 1$,
				$m_1 = m$, and
				for every $i ∈ [m]$ we have $w_i = δ^{[i]} u_1^i \bar{δ}^{[i]}$ and since $u_1^i ∈ \mathit{D}(Σ)$ also $w_i ∈ \mathit{D}(Σ)$.
				Furthermore, $w_1 \cdots w_m = δ^{[1]} u_1^1 \bar{δ}^{[1]} \cdots δ^{[m]} u_1^m \bar{δ}^{[m]} ∈ \mathit{mD}_{\text{c}}(Σ, \mathfrak{P})$ due to the cancellation rule.
		\item
			$f = [u_1, …, u_m]$ where $u_i ∈ \big\{ x^1_i,\; x^1_i δ^{[1]} \bar{δ}^{[1]},\; δ^{[1]} \bar{δ}^{[1]} x^1_i ∣ δ ∈ Δ_1\big\}$ for every $i ∈ \upto{m}$; then
				$w_i ∈ \big\{ u^1_i,\; u^1_i δ^{[1]} \bar{δ}^{[1]},\; δ^{[1]} \bar{δ}^{[1]} u^1_i ∣ δ ∈ Δ_1 \big\}$ for every $i ∈ \upto{m}$,
				$ℓ = 1$, and
				$m_1 = m$.
				Since $≡_Σ$ is a congruence relation (in particular, $≡_Σ$ respects composition), we have that $w_1, …, w_m ∈ \mathit{D}(Σ)$.
				By applying \cref{obs:mDyck_permutation_rule} $m$ times, we have that $w_1 \cdots w_m ∈ \mathit{mD}_{\text{c}}(Σ, \mathfrak{P})$.
	\end{enumerate}
	\proofpart{($⇐$)}
	If the cancellation rule is applied zero times in order to reduce $w_1 \cdots w_m$ to $ε$ then $w_1 = … = w_m = ε$.
	The rule $A_m → [ε, …, ε]()$ in $G_D$ clearly derives $(w_1, …, w_m)$.
	If the cancellation rule is applied $i + 1$ times in order to reduce $w_1 \cdots w_m$ to $ε$ then $w_1 \cdots w_m$ has the form $u_0σ_1v_1\widebar{σ_1}u_1 \cdots σ_ℓ v_ℓ \widebar{σ_ℓ} u_ℓ$ for some $u_0, …, u_ℓ ∈ \mathit{D}(Σ)$, $v_1, …, v_ℓ ∈ \mathit{D}(Σ)$, and $\{σ_1, …, σ_ℓ\} ∈ \mathfrak{P}$ with $v_1 \cdots v_ℓ ≡_{Σ, \mathfrak{P}} ε$.
	Then we need to apply the cancellation rule at most $i$ times to reduce $v_1 \cdots v_ℓ$ to $ε$, hence, by induction hypothesis, there is some $d ∈ D_{G_Δ}$ that derives $(v_1, …, v_ℓ)$.
	We use an appropriate rule $ρ$ of type (ii) such that $ρ(d)$ derives $(σ_1v_1\widebar{σ_1}, …, σ_ℓ v_ℓ \widebar{σ_ℓ})$.
	Also, we need to apply the cancellation rule at most $i$ times in order to reduce $u_0 \cdots u_ℓ$ to $ε$, hence, by induction hypothesis, there are derivations $d_1, …, d_n$ that derive tuples containing exactly $u_0, …, u_ℓ$ as components.
	Then there is a rule $ρ'$ of type \ref{def:Dyck-MCFG:3} such that $ρ'(ρ(d), d_1, …, d_n) ∈ D_{G_Δ}$ derives the tuple $(w_1, …, w_m)$.
	
	From \eqref{eq:theIH} with $m = 1$ and the fact “$w_1 ∈ \mathit{mD}_{\text{c}}(Σ, \mathfrak{P})$ implies $w_1 ∈ \mathit{D}(Σ)$” we get that $\mathit{mD}_{\text{c}}(Σ, \mathfrak{P}) = \mathit{mD}$.
%
\end{proof}

\begin{lemma}\label{thm:multipleDyckIsMCFL}
	\enspace$\cmDYCK[k] ⊆ \MCFL[k]$
\end{lemma}
\begin{proofidea}
	For a given congruence multiple Dyck language $L$, we construct a multiple Dyck grammar that is equivalent up to a homomorphism~$g$.
	We then use the closure of $\MCFL[k]$ under homomorphisms.
\end{proofidea}
\begin{proof}
	Let $L ∈ \cmDYCK[k]$.
	Then there are an alphabet $Σ$ and a partition $\mathfrak{P}$ of $Σ$ such that $\mathit{mD}_{\text{c}}(Σ, \mathfrak{P}) = L$.
	Consider $\mathfrak{P}$ as an $ℕ$-sorted set where the sort of an element is its cardinality.
	Then $\widehat{Δ} = \{\mathfrak{p}^{[i]}, \bar{\mathfrak{p}}^{[i]} ∣ \mathfrak{p} ∈ \mathfrak{P},  i ∈ \upto{\lvert\mathfrak{p}\rvert} \}$.
	For every $\mathfrak{p} ∈ \mathfrak{P}$ assume some fixed enumeration of the elements of $\mathfrak{p}$.
	We define a bijection $g: \widehat{Δ} → Σ ∪ \widebar{Σ}$ such that every $\mathfrak{p}^{[i]}$ (for some $\mathfrak{p}$ and $i$) is assigned the $i$-th element of $\mathfrak{p}$ and $g(\bar{\mathfrak{p}}^{[i]}) = \widebar{g(\mathfrak{p}^{[i]})}$.
	
	Let $\operatorname{Tup}_g(G_{\mathfrak{P}})$ denote the set of tuples obtained by interpreting the terms corresponding to every subderivation in $G_{\mathfrak{P}}$ and then applying $g$ to every component.
	We show the following claim by induction:
	\begin{align*}
		w ∈ \mathit{mD}(\mathfrak{P})
		&\iff
		∀
			ℓ ∈ ℕ,
			u_0, …, u_ℓ,w_1, …, w_ℓ ∈ \mathit{D}(Σ) \\*
		&\qquad\quad\text{ with }
			w = u_0w_1u_1 \cdots w_ℓ u_ℓ
		: \\*
			&\qquad\qquad(u_0, w_1, u_1, \cdots, w_ℓ, u_ℓ) ∈ \operatorname{Tup}_g(G_{\mathfrak{P}})
		\tag{IH}\label{cl:multipleDyckIsMCFL:IH}
	\end{align*}
	Note that in the following the indices of the elements of $\mathfrak{p} = \{σ_1, …, σ_ℓ\}$ are chosen such that they respect the previously fixed enumeration of $\mathfrak{p}$, i.e. for every $i ∈ \upto{ℓ}: g(\mathfrak{p}^{[i]}) = σ_i$. We derive
	\begin{align*}
		&w ∈ \mathit{mD}(\mathfrak{P}) \\*[.5em]
		&\iff
			∀
			ℓ ∈ ℕ,
			u_0, …, u_ℓ, v_1, …, v_ℓ ∈ \mathit{D}(Σ),
			\mathfrak{p} = \{σ_1, …, σ_ℓ\} ∈ \mathfrak{P}\text{ with } \\*
		&\qquad\quad
			w = u_0 σ_1 v_1 \widebar{σ_1} u_1 \cdots σ_ℓ v_ℓ \widebar{σ_ℓ} u_ℓ
		: 
			\enspace u_0 v_1 u_1 \cdots v_ℓ u_ℓ ∈ \mathit{mD}(\mathfrak{P})
		\tag{by def. of $\mathit{mD}(\mathfrak{P})$} \\[.5em]
		&\iff
			∀
			ℓ ∈ ℕ,
			u_0, …, u_ℓ, v_1, …, v_ℓ ∈ \mathit{D}(Σ),
			\mathfrak{p} = \{σ_1, …, σ_ℓ\} ∈ \mathfrak{P}\text{ with } \\*
		&\qquad\quad
			w = u_0 σ_1 v_1 \widebar{σ_1} u_1 \cdots σ_ℓ v_ℓ \widebar{σ_ℓ} u_ℓ
		: 
			(u_0, v_1, u_1, …, v_ℓ, u_ℓ) ∈ \operatorname{Tup}_g(G_{\mathfrak{P}})
		\tag{by \eqref{cl:multipleDyckIsMCFL:IH}} \\[.5em]
		&\iff
			∀
			ℓ ∈ ℕ,
			u_0, …, u_ℓ, v_1, …, v_ℓ ∈ \mathit{D}(Σ),
			\mathfrak{p} = \{σ_1, …, σ_ℓ\} ∈ \mathfrak{P} \text{ with } \\*
			&\qquad\quad
			w = u_0 σ_1 v_1 \widebar{σ_1} u_1 \cdots σ_ℓ v_ℓ \widebar{σ_ℓ} u_ℓ
		: \\*
			&\qquad\qquad (u_0, σ_1 v_1 \widebar{σ_1}, u_1, …, σ_ℓ v_ℓ \widebar{σ_ℓ}, u_ℓ) ∈ \operatorname{Tup}_g(G_{\mathfrak{P}})
		\tag{by def. of $G_{\mathfrak{P}}$} \\*[.5em]
		&\iff
			∀
				ℓ ∈ ℕ,
				u_0, …, u_ℓ, w_1, …, w_ℓ ∈ \mathit{D}(Σ) \text{ with }
				w = u_0 w_1 u_1 \cdots w_ℓ u_ℓ
		: \\*
			&\qquad\qquad (u_0, w_1, u_0, …, w_ℓ, u_ℓ) ∈ \operatorname{Tup}_g(G_{\mathfrak{P}})
		\tag{using permuting productions in $G_{\mathfrak{P}}$}
	\end{align*}
	$g(L(G_{\mathfrak{P}})) = L$ follows by instantiating \eqref{cl:multipleDyckIsMCFL:IH} for $ℓ = 0$ and discovering that $\{ t ∣ (t) ∈ \operatorname{Tup}_g(G_{\mathfrak{P}})\} = g(L(G_{\mathfrak{P}}))$.
	
	Since $k$-MCFLs are closed under homomorphisms \cite[Theorem~3.9]{SekMatFujKas91}, we know that $L ∈ \MCFL[k]$.%
	\footnote{This construction shows that congruence multiple Dyck languages (\cref{def:mDYCK}) are equivalent to multiple Dyck languages \cite[Definition~1]{YosKajSek10} up to the application of $g$.}
\end{proof}

\begin{observation}\label{obs:(k+1)-cmDYCK_not_k-MCFL}
Examining the definition of multiple Dyck grammars, we observe that some production in \cref{def:Dyck-MCFG:2} has fan-out $k$ for at least one $δ ∈ Δ$.
Then, using Seki, Matsumura, Fujii, and Kasami~\cite[Theorem~3.4]{SekMatFujKas91}, we have for every $k ≥ 1$ that $\cmDYCK[(k+1)] ∖ \MCFL[k] ≠ ∅$.
\end{observation}

\begin{proposition}\label{prop:mDyckHierarchy}
	\enspace\( \DYCK = \cmDYCK[1] ⊊ \cmDYCK[2] ⊊ \dots \)
\end{proposition}
\begin{proof}
	We get ‘$⊆$’ from the definition of $\cmDYCK[k]$ and ‘$≠$’ from \cref{obs:(k+1)-cmDYCK_not_k-MCFL}.
	We have the equality since the dimension of some partition $\mathfrak{P}$ of $Σ$ is 1 if and only if $\mathfrak{P} = \{\{σ\} ∣ σ ∈ Σ\}$.
	Then we have ${≡_Σ} = {≡_{Σ, \mathfrak{P}}}$ and thus $\mathit{D}(Σ) = \mathit{mD}_{\text{c}}(Σ, \mathfrak{P})$.
	Hence $\DYCK = \cmDYCK[1]$.
\end{proof}


\subsection{Membership in a congruence multiple Dyck language}

We provide a recursive algorithm (\cref{alg:IsMultipleDyck}) to decide whether a word $w$ is in a given congruence multiple Dyck language $\mathit{mD}_{\text{c}}(Σ, \mathfrak{P})$.
This amounts to checking whether $w ≡_{Σ, \mathfrak{P}} ε$, and it suffices to only apply the cancellation rule from left to right.

In order for \cref{alg:IsMultipleDyck} to decide the membership in a multiple Dyck language it must consider all decompositions of the input string into Dyck words.
For this purpose we define a function \textproc{split} that decomposes a given Dyck word into \emph{shortest, non-empty} Dyck words.

\subsubsection*{The function {\normalfont\textproc{split}}}

As Dyck languages are recognisable by pushdown automata, we define a data structure pushdown and two functions with side-effects on pushdowns.
A \emph{pushdown} is a string over some alphabet $Γ$.
Let $Γ$ be an alphabet, $γ ∈ Γ$, and $\mathit{pd} ⊆ Γ^*$ be a pushdown.
\begin{itemize}
	\item
		$\pdpop{\mathit{pd}}$ returns the left-most symbol of $\mathit{pd}$ and removes it from $\mathit{pd}$.
	\item
		$\pdpush{\mathit{pd}}{γ}$ prepends $γ$ to $\mathit{pd}$.
\end{itemize}
Note that $\pdpop{}$ is only a partial function, it is undefined for $\mathit{pd} = ε$.
But since the input word $w$ is required to be in $\mathit{D}(Σ)$ by \cref{alg:split}, the expression on line~6 is always defined.

\begin{algorithm}[ht]
	\caption{Algorithm to split a word in $\mathit{D}(Σ)$ into shortest non-empty strings from $\mathit{D}(Σ)$}
	\label{alg:split}
	\begin{algorithmic}[1]
		\Require
			alphabet $Σ$,\enspace
			Dyck word $w ∈ \mathit{D}(Σ)$
		\Ensure
			sequence $(u_1, …, u_ℓ)$ of shortest, non-empty Dyck words with $w = u_1 \cdots u_ℓ$
		\vspace{.5em}
		\Function{split}{$Σ, w$}
			\State let\; $\mathit{pd} = ε$,\; $j = 1$,\; and\; $u_j = ε$
			\For{$0 ≤ i ≤ \lvert w \rvert$}
				\State append $w_i$ to $u_j$
				\If{$w_i ∈ \widebar{Σ}$}
					\State $\pdpop{\mathit{pd}}$
					\If{$\mathit{pd} = ε$}
						\State increase $j$ by 1\; and\; let $u_j = ε$
					\EndIf
				\Else
					\State $\pdpush{\mathit{pd}}{\widebar{w_i}}$
				\EndIf
			\EndFor
			\State \Return $(u_1, …, u_{j-1})$
		\EndFunction
	\end{algorithmic}
\end{algorithm}

One can easily see that \textproc{split} is bijective (the inverse function is concatenation).
We therefore say that $w$ and $(u_1, …, u_ℓ)$ correspond to each other, and for every operation on either of them there is a corresponding operation on the other.
In particular the empty string corresponds to the empty tuple.

\subsubsection*{Outline of the function {\normalfont\textproc{isMember}}}

\algnewcommand{\IIf}[1]{\State\algorithmicif\ #1\ \algorithmicthen}
\algnewcommand{\EndIIf}{\unskip\ \algorithmicend\ \algorithmicif}

\algnewcommand{\FFor}[1]{\State\algorithmicfor\ #1\ \algorithmicdo}
\algnewcommand{\EndFFor}{\unskip\ \algorithmicend\ \algorithmicfor}

\begin{algorithm*}[t]
	\caption{Function \textproc{isMember} to decide membership in $\mathit{mD}_{\text{c}}(Σ, \mathfrak{P})$}
	\label{alg:IsMultipleDyck}
  \begin{algorithmic}[1]
		\Require
			$Σ$,
			$\mathfrak{P}$, and
			$w ∈ (Σ ∪ \widebar{Σ})^*$
		\Ensure
			1 if $w ∈ \mathit{mD}_{\text{c}}(Σ, \mathfrak{P})$,\enspace
			0 otherwise
		\vspace{1em}
		\Function{isMember}{$Σ, \mathfrak{P}, w$}
      \IIf{$w = ε$}
        \enspace\Return{1}\enspace
      \EndIIf
			\IIf{$w \not∈ \mathit{D}(Σ)$}
				\enspace\Return{0}\enspace
			\EndIIf
			\State $(σ_1u_1\widebar{σ}_1, …, σ_{\ell}u_{\ell}\widebar{σ}_{\ell}) \gets \Call{split}{Σ, w}$
      \Comment{such that $σ_1, …, σ_{\ell} ∈ Σ$}
      \State $\mathcal{I} \gets \{I ⊆ \mathcal{P}([\ell]) \mid \text{$I$ partition of $[\ell]$}, \forall\{i_1, …, i_k\} ∈ I\colon \{σ_{i_1}, …, σ_{i_k}\} ∈ \mathfrak{P}\}$
      \For{$I ∈ \mathcal{I}$}
        \State $b \gets 1$
        \For{$\{i_1, …, i_k\} ∈ I$}
        \Comment{such that $i_1 < … < i_k$}
          \State $b \gets b \cdot \Call{isMember}{Σ, \mathfrak{P}, u_{i_1} \cdots u_{i_k}}$
        \EndFor
        \IIf{$b = 1$}
          \enspace\Return{1}\enspace
        \EndIIf
      \EndFor
      \State\Return{0}
    \EndFunction
  \end{algorithmic}
\end{algorithm*}

If $w$ is the empty word, we return 1 on line~2 since the empty word is in $\mathit{mD}_{\text{c}}(Σ, \mathfrak{P})$.
Then we check if $w$ is in $\mathit{D}(Σ)$, e.g. with the context-free grammar in (7.6) in Salomaa~\cite{sal73}.
If $w$ is not in $\mathit{D}(Σ)$, it is also not in $\mathit{mD}_{\text{c}}(Σ, \mathfrak{P})$ and we return 0.
Otherwise, we split $w$ into shortest non-empty Dyck words (on line~4) using the function \textproc{split}.
Since each of those shortest non-empty Dyck words has the form $σu\widebar{σ}$ for some $σ ∈ Σ$ and $u ∈ (Σ \cup \widebar{Σ})^*$, we write $(σ_1u_1\widebar{σ}_1, …, σ_{\ell}u_{\ell}\widebar{σ}_{\ell})$ for the left-hand side of the assignment on line~4.
On line~5 we calculate the set $\mathcal{I}$ of all partitions $I$ such that each element of $I$ specifies a set of components of the tuple $(σ_1u_1\widebar{σ}_1, …, σ_{\ell}u_{\ell}\widebar{σ}_{\ell})$ whose outer parentheses can be removed with one application of the cancellation rule.
Since $I$ is a partition, we know that the outer parentheses of every component can be removed via the cancellation rule.
Then it remains to be shown that there is a partition $I$ such that for each element $\{i_1, …, i_k\}$ of $I$ the word $u_{i_1} \cdots u_{i_k}$ is an element of $\mathit{mD}_{\text{c}}(Σ, \mathfrak{P})$; this is done on lines~6 to~12. 
If there is no such partition, then we return 0 on line~13.

\begin{table}[t]
  \caption{Run of \cref{alg:IsMultipleDyck} on the word $⟦()⟧[⟨⟩]$, cf. \cref{ex:multipleDyck,ex:IsMultipleDyck}.}
	\label{tab:runIsMultipleDyck}
  \begin{tabbing}
    l.\,11: \= $b = 1 \cdot {}$ \= l.\,11: \= $b = 1 \cdot {}$ \= l.\,4: \= \kill
    $\textproc{isMember}(Σ, \mathfrak{P}, ⟦()⟧[⟨⟩])$ \\
    l.\,4:  \> $σ_1 = {⟦}, σ_2 = {[}, u_1 = {(}{)}, u_2 = {⟨}{⟩}$ \\
    l.\,5:  \> $\mathcal{I} = \big\{ \big\{ \{1,2\} \big\} \big\}$ \\
    l.\,6:  \> $I = \big\{ \{1,2\} \big\}$ \\
    l.\,8:  \> $k = 2, i_1 = 1, i_2 = 2$ \\
    l.\,9:  \> $b = 1 \cdot \textproc{isMember}(Σ, \mathfrak{P}, {(}{)}{⟨}{⟩})$ \\
            \>           \> l.\,4:  \> $σ_1 = {(}, σ_2 = {⟨}, u_1 = ε = u_2$ \\
            \>           \> l.\,5:  \> $\mathcal{I} = \big\{ \big\{ \{1,2\} \big\} \big\}$ \\
            \>           \> l.\,6:  \> $I = \big\{ \{1,2\} \big\}$ \\
            \>           \> l.\,8:  \> $k = 2, i_1 = 1, i_2 = 2$ \\
            \>           \> l.\,9:  \> $b = 1 \cdot \textproc{isMember}(Σ, \mathfrak{P}, ε)$ \\
            \>           \>         \>               \> l.\,2: \> \textbf{return}~1 \\
            \>           \> l.\,9:  \> $b = 1 \cdot 1 = 1$ \\
            \>           \> l.\,11: \> \textbf{return}~1 \\
    l.\,9:  \> $b = 1 \cdot 1$ \\
    l.\,11: \> \textbf{return}~1
  \end{tabbing}
\end{table}

\begin{table}
  \caption{Run of \cref{alg:IsMultipleDyck} on the word $⟦()⟧[]⟦⟧[⟨⟩]$.}
  \label{tbl:ex2-main-alg}
  \begin{tabbing}
    l.\,11: \= $b = 1 \cdot {}$ \= l.\,11: \= $b = 1 \cdot {}$ \= l.\,4: \= \kill
    $\textproc{isMember}(Σ, \mathfrak{P}, ⟦()⟧[]⟦⟧[⟨⟩])$ \\
    l.\,4:  \> $σ_1= {⟦} = σ_3, σ_2 = {[} = σ_4, u_1 = (), u_2 = ε = u_3, u_4 = ⟨⟩$ \\
    l.\,5:  \> $\mathcal{I} = \big\{ \big\{ \{1,2\}, \{3,4\} \big\}, \big\{ \{1,4\}, \{2,3\} \big\} \big\}$ \\[.7em]
    l.\,6:  \> $I = \big\{ \{1,2\}, \{3,4\} \big\}$ \\
    l.\,8:  \> $k = 2, i_1 = 1, i_2 = 2$ \\
    l.\,9:  \> $b = 1 \cdot \textproc{isMember}(Σ, \mathfrak{P}, ())$ \\
            \>           \> l.\,4:  \> $σ_1 = {(}, u_1 = ε$ \\
            \>           \> l.\,5:  \> $\mathcal{I} = \emptyset$ \\
            \>           \> l.\,13: \> \textbf{return}~0 \\
    l.\,9:  \> $b = 1 \cdot 0 = 0$ \\
    l.\,8:  \> $k = 2, i_1 = 3, i_2 = 4$ \\
    l.\,9:  \> $b = 0 \cdot \textproc{isMember}(Σ, \mathfrak{P}, ⟨⟩)$ \\
            \>           \> l.\,4:  \> $σ_1 = {⟨}, u_1 = ε$ \\
            \>           \> l.\,5:  \> $\mathcal{I} = \emptyset$ \\
            \>           \> l.\,13: \> \textbf{return}~0 \\
    l.\,9:  \> $b = 0 \cdot 0 = 0$ \\[.7em]
    l.\,6:  \> $I = \big\{ \{1,4\}, \{2,3\} \big\}$ \\
    l.\,8:  \> $k = 2, i_1 = 1, i_2 = 4$ \\
    l.\,9:  \> $b = 1 \cdot \textproc{isMember}(Σ, \mathfrak{P}, ()⟨⟩)$ \\
            \>           \> l.\,4:  \> $σ_1 = {(}, σ_2 = {⟨}, u_1 = ε = u_2$ \\
            \>           \> l.\,5:  \> $\mathcal{I} = \big\{ \big\{ \{1,2\} \big\} \big\}$ \\
            \>           \> l.\,6:  \> $I = \big\{ \{1,2\} \big\}$ \\
            \>           \> l.\,8:  \> $b = 1 \cdot \textproc{isMember}(Σ, \mathfrak{P}, ε)$ \\
            \>           \>         \>          \> l.\,2: \> \textbf{return}~1 \\
            \>           \> l.\,8:  \> $b = 1 \cdot 1$ \\
            \>           \> l.\,11:  \> \textbf{return}~1 \\
    l.\,9:  \> $b = 1 \cdot 1$ \\
    l.\,8:  \> $k = 2, i_1 = 2, i_2 = 3$ \\
    l.\,9:  \> $b = 1 \cdot \textproc{isMember}(Σ, \mathfrak{P}, ε)$ \\
            \>           \> l.\,2: \> \textbf{return}~1 \\
    l.\,9:  \> $b = 1 \cdot 1$ \\[.7em]
    l.\,11: \> \textbf{return}~1
  \end{tabbing}
\end{table}

\begin{example}[\cref{ex:multipleDyck} continued]\label{ex:IsMultipleDyck}
	\Cref{tab:runIsMultipleDyck} shows a run of \cref{alg:IsMultipleDyck} on the word $⟦()⟧[⟨⟩]$ where we report return values and a subset of the variable assignment whenever we reach the end of lines~4, 5, 6, 8, 9.
	The recursive calls to $\textproc{isMember}$ are indented.
	\Cref{tbl:ex2-main-alg} shows the run of \cref{alg:IsMultipleDyck} on the word $⟦()⟧[]⟦⟧[⟨⟩]$.
\end{example}

In light of the close link between \cref{alg:IsMultipleDyck} and the relation $≡_{Σ, \mathfrak{P}}$ we omit the proof of correctness.

\begin{proof}[Proof of termination for \cref{alg:IsMultipleDyck}]
	If $w = ε$, the algorithm terminates on line~2.
	If $w \not∈\mathit{D}(Σ)$, the algorithm terminates on line~3.
  Since $\mathcal{I}$ is finite and each element $I ∈ \mathcal{I}$ is also finite, there are only finitely many calls to $\textproc{isMember}$ on line~9 for each recursion.
  In each of those calls, the length of the third argument is strictly smaller then the length of $w$.
  Therefore, after a finite number of recursions, the third argument passed to $\textproc{isMember}$ is either the empty word, then the algorithm terminates on line~2, or not an element of $\mathit{D}(Σ)$, then the algorithm terminates on line~2.
\end{proof}

\section{CS theorem for weighted MCFLs}
\label{sec:weighted_CS}

In this section we generalise the CS representation of (unweighted) MCFLs \cite[Theorem~3]{YosKajSek10} to the weighted case.
We prove that an $\mathcal{A}$-weighted MCFL $L$ can be decomposed into an $\mathcal{A}$-weighted alphabetic homomorphism $h$, a regular language $R$ and a congruence multiple Dyck language $\mathit{mD}_{\text{c}}$ such that $L = h(R ∩ \mathit{mD}_{\text{c}})$.

To show this, we use the proof idea from Droste and Vogler~\cite{DroVog13}:
We separate the weight from our grammar formalism and then use the unweighted CS representation on the unweighted part.
The outline of our proof is as follows:
\begin{enumerate}
	\item
		We separate the weights from $L$ (\cref{lem:weightSeparation}), obtaining an MCFL $L'$ and a weighted alphabetic homomorphism.
	\item
		We use a corollary of the CS representation of (unweighted) MCFLs (\cref{thm:CSForMCFL}) to obtain a CS representation of $L'$.
	\item
		Using the two previous points and a lemma for the composition of weighted and unweighted alphabetic homomorphisms (\cref{lem:HOM(S)_before_HOM=HOM(S)}), we obtain a CS representation of $L$ (\cref{thm:weightedCSForMCFL}).
\end{enumerate}

\begin{figure}
\centering
\begin{tikzpicture}[every node/.style={anchor=base,text height=.8em,text depth=.2em}, >=stealth']
	\matrix (m) [matrix of math nodes, column sep=5em, row sep=2em, inner ysep=.2em, inner xsep=.5em] {
		Δ^* → \mathcal{A} & Γ^*   & (Σ ∪ \bar{Σ})^* \\
		L(G)              & L(G_{\mathbb{B}}) & R(G_{\mathbb{B}}) ∩ \mathit{mD}(G_{\mathbb{B}}) \\[1em]
		D_G               & D_{G_{\mathbb{B}}} \\
	};
	\draw [->] (m-1-2) to node [above] {$\wts{G}$} (m-1-1);
	\draw [->] (m-1-3) to node [above] {$\hom{G_{\mathbb{B}}}$} (m-1-2);
	\draw [->, bend angle=30, bend right] (m-1-3) to node (h) [above] {$h$} (m-1-1);
	\draw [->] (m-2-2) to node [above] {$\wts{G}$} (m-2-1);
	\draw [->] (m-2-2) to node [right, xshift=2pt, very near end] {\textit{toDeriv}} (m-3-1);
	\draw [->] (m-2-3) to node [above] {$\hom{G_{\mathbb{B}}}$} (m-2-2);
	\draw [->] (m-2-3.south) to node [below, sloped] {$\fromBrackets$} (m-3-2.east);
	\draw [->] (m-3-1) to node [right, near end] {$(\yield, \hat{μ})$} (m-2-1);
	\draw [->] (m-3-1) to node [below] {$f$} (m-3-2);
	\draw [->] (m-3-2) to node [left] {$\yield$} (m-2-2);
	\draw [->] (m-3-2) to node [above, sloped] {$\toBrackets$} (m-2-3);
	\path (m-2-3) -- node [sloped] {$⊆$} (m-1-3);
	\path (m-2-2) -- node [sloped] {$⊆$} (m-1-2);
	\path (m-2-1) -- node [sloped] {$∈$} (m-1-1);

	\node (lhom) [fit=(m-1-1) (m-1-3) (h), fill=blue, fill opacity=.15, draw=blue, densely dashed, inner sep=.3em, rounded corners] {};
	\node (cCS) [fit=(m-2-3) (m-3-2), fill=red, fill opacity=.15, draw=red, densely dashed, inner sep=.6em, rounded corners] {};
	\node (lWS) [fit=(m-2-2) (m-3-1), fill=olive, fill opacity=.15, draw=olive, densely dashed, inner sep=.6em, rounded corners] {};

	\node [below=0pt of lWS.south west, anchor=north west] {\cref{lem:weightSeparation}};
	\node [below=0pt of cCS.south east, anchor=north east] {\cref{thm:CSForMCFL,cor:CSbijection}};
	\node [below=0pt of lhom.north east, anchor=south east] {\cref{lem:HOM(S)_before_HOM=HOM(S)}};
\end{tikzpicture}
\caption{Outline of the proof of \cref{thm:weightedCSForMCFL}}
\label{fig:weightedCSoutline}
\end{figure}

\Cref{fig:weightedCSoutline} outlines the proof of \cref{thm:weightedCSForMCFL}.
The boxes represent sub-diagrams for which the corresponding lemma proofs existence of the arrows and commutativity.

\subsection{Separating the weights}

We split a given weighted MCFG $G$ into an unweighted MCFG $G_{\mathbb{B}}$ and a weighted homomorphism $\wts{G}$ such that $\sem{G} = \wts{G}(L(G_{\mathbb{B}}))$.

\begin{definition}\label{def:boolean-part}
	Let $G = (N, Δ, S, P, μ)$ be a non-deleting $\mathcal{A}$-weighted $k$-MCFG.
	The \emph{unweighted MCFG for $G$} is the non-deleting $k$-MCFG $G_{\mathbb{B}} = (N, Γ, S, P')$ where $Γ = Δ ∪ \{ ρ^i ∣ ρ ∈ P, i ∈ \upto{\fanout(ρ)}\}$ and $P'$ is the smallest set such that for every production $ρ = A → [u_1, …, u_s](A_1, …, A_m) ∈ P$ there is a production \[A → [ρ^1 u_1, …, ρ^s u_s](A_1, …, A_m) ∈ P'\text{.}\qedhere\]
\end{definition}

\begin{definition}\label{def:weights-part}
	Let $G = (N, Δ, S, P, μ)$ be a non-deleting $\mathcal{A}$-weighted MCFG.
	The \emph{weight homomorphism for $G$} is the $\mathcal{A}$-weighted alphabetic homomorphism $\wts{G}: Γ^* → (Δ^* → \mathcal{A})$ where
		$\wts{G}(δ) = 1.δ$,
		$\wts{G}(ρ^1) = μ(ρ).ε$, and
		$\wts{G}(ρ^i) = 1.ε$ for every $δ ∈ Δ$, $ρ ∈ P$ and $i ∈ \{2, …, \fanout(ρ)\}$.
\end{definition}

$L(G_{\mathbb{B}})$ stands in bijection to $D_G$ via the function \textproc{toDeriv} given in \cref{alg:invYield}.

\begin{algorithm}[ht]
\caption{Function $\textproc{toDeriv}$ to calculate for every word in $L(G_{\mathbb{B}})$ the corresponding derivation in $D_G$, cf. \cref{lem:weightSeparation}}
\label{alg:invYield}
\begin{algorithmic}[1]
	\Require{$w ∈ L(G_{\mathbb{B}})$}
	\Ensure{derivation tree $t ∈ D_G$ corresponding to $w$ (represented as a partial function from $ℕ^*$ to $P$)}
	\vspace{.5em}
	\Function{toDeriv}{$w$}
	\State let $t$ be the empty function
	\State \Call{descend}{$t, ε, 1$}
	\State\Return $t$
	\EndFunction
	\vspace{.5em}
	\Procedure{descend}{$t:ℕ^* → P, π ∈ ℕ^*, j ∈ ℕ$}
		\State let $ρ = A → [u_1, …, u_s](A_1, …, A_k) ∈ P$ and $u$ such that $ρ^ju = w$
		\State add the assignment $π ↦ ρ$ to $t$
		\State remove $ρ^j$ from the beginning of $w$
		\For{\textbf{every} symbol $δ'$ in $u_j$}
			\If{$δ' ∈ Δ$}
				\State remove $δ'$ from the beginning of $w$
			\Else
				\State let $i, j'$ such that $x_i^{j'} = δ'$
				\State\Call{descend}{$t, πi, j'$}
			\EndIf
		\EndFor
	\EndProcedure
\end{algorithmic}
\end{algorithm}

\begin{lemma}\label{lem:weightSeparation}
	\enspace\(\MCFL[k](\mathcal{A}) = \aHOM[\mathcal{A}] \big(\MCFL[k]\big)\)
\end{lemma}
\begin{proof}
	\proofpartns{($⊆$)} Let $L ∈ \MCFL[k](\mathcal{A})$.
	By \cref{lem:normalform} there is a non-deleting $\mathcal{A}$-weighted $k$-MCFG $G = (N, Δ, S, P, μ)$ such that $\sem{G} = L$.
	Let $f$ be the function obtained by applying the construction of the rules in $G_{\mathbb{B}}$ position-wise to a derivation in $D_G$.
	For every $w ∈ L(G_{\mathbb{B}})$ we can calculate the corresponding derivation $t$ in $G$ (as a function with domain $\dom(t)$ and labelling function $t$) using \textproc{toDeriv} (\cref{alg:invYield}), hence $\yield ∘ f$ is bijective.
	We derive for every $w ∈ Δ^*$:
	\begin{align*}
		L(w)
		&= \sem{G}(w) \\*
		&= \textstyle∑_{d ∈ D_G(w)} μ(d) \\
		&= \textstyle∑_{d ∈ D_G} (\wts{G} ∘ \yield ∘ f)(d)(w)
			\tag{by $\dagger$} \\
		&= \textstyle∑_{\substack{d ∈ D_G, u ∈ L(G_{\mathbb{B}}) \\ u = (\yield ∘ f)(d)}} \wts{G}(u)(w) \\
		&= \textstyle∑_{u ∈ L(G_{\mathbb{B}})} \wts{G}(u)(w)
			\tag{$L(G_{\mathbb{B}})$ and $D_G$ are in bijection} \\*
		&= \wts{G}(L(G_{\mathbb{B}}))(w)
	\end{align*}
	For $\dagger$, one can immediately see from the definitions of $f$, $\yield$, and $\wts{G}$ that for every $w ∈ Δ^*$ we have $(\wts{G} ∘ \yield ∘ f)(d)(w) = μ(d)$ if $d ∈ D_G(w)$ and $(\wts{G} ∘ \yield ∘ f)(d)(w) = 0$ otherwise.
	Hence $L = \wts{G}(L(G_{\mathbb{B}}))$.
	
	\proofpart{($⊇$)}
	Let $L ∈ \MCFL[k]$ and $h: Γ^* → (Δ^* → \mathcal{A})$ be an $\mathcal{A}$-weighted alphabetic homomorphism.
	By Seki, Matsumura, Fujii, and Kasami~\cite[Lemma~2.2]{SekMatFujKas91} there is a non-deleting $k$-MCFG $G = (N, Γ, S, P)$ such that $L(G) = L$.
	We construct the $\mathcal{A}$-weighted $k$-MCFG $G' = (N, Δ, S, P', μ)$ as follows:
	We extend $h$ to $h': (Γ ∪ X)^* → ((Δ ∪ X)^* → \mathcal{A})$ where
		$h'(x) = 1.x$ for every $x ∈ X$ and
		$h'(γ) = h(γ)$ for every $γ ∈ Γ$.
	We define $P'$ as the smallest set such that
		for every
			$ρ = A → [u_1, …, u_s](A_1, …, A_m) ∈ P_{(s_1 \cdots s_m, s)}$ and
			$(u_1', …, u_s') ∈ \supp(h'(u_1)) × … × \supp(h'(u_s))$
		we have that $P'$ contains the production $ρ' = A → [u_1', …, u_s'](A_1, …, A_m)$ and
		$μ(ρ') = h'(u_1)(u_1') ⋅ … ⋅ h'(u_s)(u_s')$.
	Since ${⋅}$ is commutative and $G$ is non-deleting, we obtain $\sem{G'} = h(L(G))$.%
\end{proof}

By setting $k = 1$ in the above lemma we reobtain the equivalence of 1 and 3 in Theorem~2 of Droste and Vogler~\cite{DroVog13} for the case of complete commutative strong bimonoids.

\begin{example}\label{ex:weightSeparation}
	Recall the $\mathrm{Pr}_2$-weighted MCFG $G$ from \cref{ex:wmcfg}.
	By \cref{def:boolean-part,def:weights-part} we obtain
	the MCFG $G_{\mathbb{B}} = \big( N, Γ, S, P' \big)$
	where
		\[Γ = \{a, b, c, d, ρ_1^1, ρ_2^1, ρ_2^2, ρ_3^1, ρ_3^2, ρ_4^1, ρ_4^2, ρ_5^1, ρ_5^2 \}\] and
		$P'$ is given by
	\begin{align*}
		P': \enspace
		ρ_1' &= S → [ρ_1^1 x_1^1 x_2^1 x_1^2 x_2^2](A, B) \\*
		ρ_2' &= A → [ρ_2^1 a x_1^1, ρ_2^2 c x_1^2](A)
		&ρ_4' &= A → [ρ_4^1, ρ_4^2]() \\*
		ρ_3' &= B → [ρ_3^1 b x_1^1, ρ_3^2 d x_1^2](B)
		&ρ_5' &= B → [ρ_5^1, ρ_5^2]()\text{,}
	\end{align*} and
	the $\mathcal{A}$-weighted alphabetic homomorphism $\wts{G}: Γ^* → (Δ^* → \mathcal{A})$ where $\wts{G}$ is given for every $γ ∈ Γ$ and $ω ∈ Δ ∪ \{ε\}$ by
	\[
		\wts{G}(γ)(ω) =
		\begin{cases}
			μ(ρ_i) &\text{if $γ = ρ_i^1$ and $ω = ε$ for $1 ≤ i ≤ 5$} \\
			1 &\text{if $γ = ρ_i^2$ and $ω = ε$ for $2 ≤ i ≤ 5$} \\
			1 &\text{if $γ ∈ Δ$ and $ω = γ$} \\
			0 &\text{otherwise,}
		\end{cases}
	\]
	Now consider the (only) derivation $d = ρ_1\big(ρ_2(ρ_4), ρ_5\big)$ of $w = ac$.
	Then
	\begin{align*}
		f(d)
			&= ρ_1'\big(ρ_2'(ρ_4'), ρ_5'\big)
			&&=: d'\text{,} \\
		g(d')
			&= ρ_1^1\; ρ_2^1\; a\; ρ_4^1\; ρ_5^1\; ρ_2^2\; c\; ρ_4^2\; ρ_5^2
			&&=: w'\text{, and} \\
		\wts{G}(w')
			&= (1 ⋅ 1/2 ⋅ 1 ⋅ 1/2 ⋅ 2/3 ⋅ 1 ⋅ 1 ⋅ 1 ⋅ 1).w
			&&= (1/6).w\;\text{.} \qedhere
	\end{align*}
\end{example}

\subsection{Strengthening the unweighted CS representation}

Yoshinaka, Kaji, and Seki~\cite{YosKajSek10} define (in Section~3.2) an indexed alphabet $Δ$, a right-linear regular grammar $R$, and a homomorphism $h$ for some given non-deleting MCFG $G$ that has no rule with at least two identical non-terminals on the right-hand side.
We will sometimes write $\hom{G}$ instead of $h$ to highlight the connection to $G$.
Let $S$ be the initial non-terminal of $G$.
Then $R$ can be viewed as an automaton $ℳ(G)$ by setting $S$ as the initial state, $T$ as final state, and having for every rule $A → wB$ in $R$ (where $A$ and $B$ are non-terminals and $w$ is a terminal string) a transition $(A, w, B)$ in $ℳ(G)$.
For every MCFG $G$ with the above three properties we define the \emph{generator automaton with respect to $G$} as $ℳ(G)$ and the \emph{generator language with respect to $G$} as $R(G) = L(ℳ(G))$.
Note that $ℳ(G)$ is deterministic.
We call $\widehat{Δ}$ the \emph{generator alphabet with respect to $G$}.
By \cref{prop:congruence_mDyck} we know that $L(D_Δ)$ \cite[Definition~1]{YosKajSek10} is a congruence multiple Dyck language, we will denote it by $\mathit{mD}(G)$ to highlight its connection to the MCFG $G$.
For every terminal symbol $γ ∈ Γ$, let $\tilde{γ}$ abbreviate the string ${⟦_γ^{[1]}} {⟧_γ^{[1]}}$.
We give an example to show how the above considerations are used.
\begin{example}[{\cref{ex:wmcfg,ex:weightSeparation} continued}] \label{ex:CSForMCFL}
	\Cref{fig:checkerR} shows the FSA $ℳ(G')$.
	An edge labelled with a set $L$ of words denotes a set of transitions each reading a word in $L$.
	Note that $R(G')$ is not finite.
	Let $Σ$ be the generator alphabet with respect to $G'$
	The homomorphism $\hom{G_{\mathbb{B}}}: Σ → Γ^*$ is given by
	\[
		\hom{G_{\mathbb{B}}}(σ) =
		\begin{cases}
			γ &\text{if $σ = ⟦_γ$ for some $γ ∈ Γ$} \\
			ε &\text{otherwise.}
		\end{cases}\qedhere
	\]
\end{example}

The following is a corollary to Yoshinaka, Kaji, and Seki~\cite[Theorem~3]{YosKajSek10} where “homomorphism” is replaced by an “alphabetic homomorphism” and “multiple Dyck language” is replaced by “congruence multiple Dyck language”.

\begin{corollary}
\label{thm:CSForMCFL}
	Let $L$ be a language and $k ∈ ℕ$.
	Then the following are equivalent:
	\begin{enumerate}
		\item $L ∈ \MCFL[k]$
		\item
			there are
				an alphabetic homomorphism $h_2$,
				a regular language $R$, and
				a congruence multiple Dyck language $\mathit{mD}_{\text{c}}$ of at most dimension $k$
			with
				$L = h_2(R ∩ \mathit{mD}_{\text{c}})$.
	\end{enumerate}
\end{corollary}
\begin{proof}
	The construction of the homomorphism in Yoshinaka, Kaji, and Seki~\cite[Section~3.2]{YosKajSek10} already satisfies the definition of an alphabetic homomorphism.
	We may use a congruence multiple Dyck language instead of a multiple Dyck language since,
		for (i) $⇒$ (ii),
			$\mDYCK[k] ⊆ \cmDYCK[k]$ (\cref{prop:congruence_mDyck}) and,
		for (ii) $⇒$ (i),
			$\cmDYCK[k] ⊆ \MCFL[k]$ (\cref{thm:multipleDyckIsMCFL}) and $\MCFL[k]$ is closed under intersection with regular languages and under homomorphisms \cite[Theorem~3.9]{SekMatFujKas91}.
\end{proof}

\begin{figure}[t]
  \centering
	\begin{tikzpicture}[node distance=7em, >=stealth', bend angle=30, every loop/.append style={looseness=10}, baseline=(S1.base)]
	\node (S1) [state, initial] {$S^{[1]}$};
	\node (A1) [state, right=of S1] {$A^{[1]}$}
		edge [pre] node [above] {$⟦_{ρ_1'}^{[1]} \tilde{ρ}_1^1 ⟦_{ρ_1',1}^{[1]}$} (S1)
		edge [loop right] node [right] {$⟦_{ρ_2'}^{[1]} \tilde{ρ}_2^1 \tilde{a} ⟦_{ρ_2',1}^{[1]}$} ();
	\node (T) [state, accepting, below=of A1] {$T$}
		edge [pre] node [left, near end] {$⟦_{ρ_4'}^{[1]} \tilde{ρ}_4^1 ⟧_{ρ_4'}^{[1]}$} (A1)
		edge [bend angle=30, out=30, in=60, looseness=10,->] node [above, xshift=5em, align=left] {$\big\{ ⟧_{ρ_1',2}^{[2]} ⟧_{ρ_1}^{[1]},\; ⟧_{ρ_2',1}^{[1]} ⟧_{ρ_2'}^{[1]},$ \\ $\hphantom{\big\{}⟧_{ρ_2',1}^{[2]} ⟧_{ρ_2'}^{[2]},\; ⟧_{ρ_3',1}^{[1]} ⟧_{ρ_3'}^{[1]},\; ⟧_{ρ_3',1}^{[2]} ⟧_{ρ_3'}^{[2]} \big\}$} (T);
	\node (A2) [state] at ($(T) + (-190:12em)$) {$A^{[2]}$}
		edge [post, bend left] node [above] {$⟦_{ρ_4'}^{[2]} \tilde{ρ}_4^2 ⟧_{ρ_4'}^{[2]}$} (T)
		edge [pre, bend right] node [above] {$⟧_{ρ_1',2}^{[1]} ⟦_{ρ_1',1}^{[2]}$} (T)
		edge [loop below] node [below] {$⟦_{ρ_2'}^{[2]} \tilde{ρ}_2^2 \tilde{c} ⟦_{ρ_2',1}^{[2]}$} ();
	\node (B1) [state] at ($(T) + (-110:12em)$) {$B^{[1]}$}
		edge [post, bend left] node [left, near start] {$⟦_{ρ_5'}^{[1]} \tilde{ρ}_5^1 ⟧_{ρ_5'}^{[1]}$} (T)
		edge [pre, bend right] node [right, near start] {$⟧_{ρ_1',1}^{[1]} ⟦_{ρ_1',2}^{[1]}$} (T)
		edge [loop left] node [left] {$⟦_{ρ_3'}^{[1]} \tilde{ρ}_3^1 \tilde{b} ⟦_{ρ_3',1}^{[1]}$} ();
	\node (B2) [state] at ($(T) + (-30:12em)$) {$B^{[2]}$}
		edge [post, bend left] node [right, xshift=-1.3em, yshift=1.3em] {$⟦_{ρ_5'}^{[2]} \tilde{ρ}_5^2 ⟧_{ρ_5'}^{[2]}$} (T)
		edge [pre, bend right] node [right, yshift=.5em] {$⟧_{ρ_1',1}^{[2]} ⟦_{ρ_1',2}^{[2]}$} (T)
		edge [loop below] node [below] {$⟦_{ρ_3'}^{[2]} \tilde{ρ}_3^2 \tilde{d} ⟦_{ρ_3',1}^{[2]}$} ();
\end{tikzpicture}
	\caption{Automaton $ℳ(G_{\mathbb{B}})$ (cf. \cref{ex:CSForMCFL})}
	\label{fig:checkerR}
\end{figure}

\begin{lemma}\label{cor:CSbijection}
	For every MCFG $G$, there is a bijection between $D_G$ and $R(G) ∩ \mathit{mD}(G)$.
\end{lemma}
\begin{proof}
	The constructions in Lemmas~1 and~3 in Yoshinaka, Kaji, and Seki~\cite{YosKajSek10} already hint on the bijection between $R(G) ∩ \mathit{mD}(G)$ and $D_G$, we will merely point out the respective functions $\toBrackets: D_G → R(G) ∩ \mathit{mD}(G)$ and $\fromBrackets: R(G) ∩ \mathit{mD}(G) → D_G$ here.

	We examine the proof of Lemma~1 in Yoshinaka, Kaji, and Seki~\cite{YosKajSek10}.
	They construct for every rule $A → f(B_1, …, B_k)$ in $G$ and all tuples $\bar{u}_1, …, \bar{u}_k$ that are generated by $B_1, …, B_k$, respectively, a new tuple $\bar{u} = (u_1, …, u_m)$ such that $\bar{u} ∈ L(G_Δ)$, for each $i ∈ [m]$, $ℳ(G)$ recognises $u_i$ on the way from $A^{[i]}$ to $T$, and $\hom{G}(\bar{u}) = f(\hom{G}(τ_1), …, \hom{G}(τ_k))$, where $\hom{G}$ is applied to tuples component-wise.
	Now we only look at the initial non-terminal $S$.
	Then $\bar{u}$ has only one component and this construction can be conceived as a function $\toBrackets: D_G → R(G) ∩ \mathit{mD}(G)$ such that $\hom{G} ∘ \toBrackets = \yield$.
	
	In Lemma~3, Yoshinaka, Kaji, and Seki~\cite{YosKajSek10} give a construction for the opposite direction by recursion on the structure of derivations in $G_Δ$.
	In a similar way as above, we view this construction as a function $\fromBrackets: R(G) ∩ \mathit{mD}(G) → D_G$ such that $\yield ∘ \fromBrackets = \hom{G}$.
	Then $\hom{G} ∘ \toBrackets ∘ \fromBrackets = \hom{G}$, and hence $\toBrackets ∘ \fromBrackets$ is the identity on $R(G) ∩ \mathit{mD}(G)$.
\end{proof}

\subsection{Composing the homomorphisms}

\begin{lemma}
\label{lem:HOM(S)_before_HOM=HOM(S)}
	\enspace\( \aHOM[\mathcal{A}] ∘ \aHOM = \aHOM[\mathcal{A}] \)
\end{lemma}
\begin{proof}
	\proofpartns{($⊆$)}
	Let $h_1: Γ^* → (Δ^* → \mathcal{A})$ be an alphabetic $\mathcal{A}$-weighted homomorphism and $h_2: Σ^* → Γ^*$ be an alphabetic homomorphism.
	By the definitions of $\aHOM[\mathcal{A}]$ and $\aHOM$ there must exist $h_1': Γ → (Δ ∪ \{ε\} → \mathcal{A})$ and $h_2': Σ → Γ ∪ \{ε\}$ such that $\widehat{h_1'} = h_1$ and $\widehat{h_2'} = h_2$.
	Since $h_1(\range(h_2')) ⊆ (Δ ∪ \{ε\} → \mathcal{A})$ there is some $h ∈ \aHOM[\mathcal{A}]$ such that $h = h_1 ∘ h_2$; hence $h_1 ∘ h_2 ∈ \aHOM[\mathcal{A}]$.
	\proofpart{($⊇$)}
	Let $h: Σ → (Γ^* → \mathcal{A})$ be an alphabetic $\mathcal{A}$-weighted homomorphism.
	Clearly $i: Σ^* → Σ^*$ with $i(w) = w$ for every $w ∈ Σ^*$ is an alphabetic homomorphism.
	Then we have $h ∘ i = h$.
\end{proof}

\begin{example}[{\cref{ex:weightSeparation,ex:CSForMCFL} continued}]\label{ex:HOM(S)_before_HOM=HOM(S)}
	The homomorphism $h: (Σ ∪ \bar{Σ})^* → (Δ^* → \mathcal{A})$ obtained from $\wts{G}: Γ^* → (Δ^* → \mathcal{A})$ and $\hom{G_{\mathbb{B}}}: (Σ ∪ \bar{Σ})^* → Γ^*$ by the construction for $⊆$ in \cref{lem:HOM(S)_before_HOM=HOM(S)} is given for every $σ ∈ Σ$ and $ω ∈ Δ ∪ \{ε\}$ by
	\[
		h(σ)(ω) =
		\begin{cases}
			μ(ρ_i)
				&\text{if $σ = ⟦_{ρ_i^1}$ and $ω = ε$ for some $i ∈ [5]$} \\
			1 &\text{if $σ ∉ \{{⟦_{ρ_i^1}} ∣ i ∈ [5]\} ∪ \{ {⟦_δ} ∣ δ ∈ Δ\}$ and $ω = ε$} \\
			1 &\text{if $σ = ⟦_δ$ and $ω = δ$ for some $δ ∈ Δ$} \\
			0 &\text{otherwise.}
		\end{cases}\qedhere
	\]
\end{example}

\subsection{The weighted CS representation}

\begin{theorem}
\label{thm:weightedCSForMCFL}
	Let $L$ be an $\mathcal{A}$-weighted language over $Σ$ and $k ∈ ℕ$. The following are equivalent:
	\begin{enumerate}
		\item\label{item:weightedCSForMCFL:1}
			$L ∈ \MCFL[k](\mathcal{A})$
		\item\label{item:weightedCSForMCFL:2}
			there are
				an $\mathcal{A}$-weighted alphabetic homomorphism $h$,
				a regular language $R$, and
				a congruence multiple Dyck language $\mathit{mD}$ of dimension at most $k$ with
				\( L = h(R ∩ \mathit{mD})\).
	\end{enumerate}
\end{theorem}
\begin{proof}
	\proofpartns{\ref{item:weightedCSForMCFL:1} $⇒$ \ref{item:weightedCSForMCFL:2}}
	There are some $L' ∈ \MCFL[k]$, $h, h_1 ∈ \aHOM[\mathcal{A}]$, $h_2 ∈ \aHOM$, $\mathit{mD} ∈ \cmDYCK[k]$, and $R ∈ \REG$ such that
	\begin{align*}
		L
		&= h_1(L')
			\tag{by \cref{lem:weightSeparation}} \\*
		&= h_1(h_2(R ∩ \mathit{mD}))
			\tag{by \cref{thm:CSForMCFL}} \\*
		&= h(R ∩ \mathit{mD})
			\tag{by \cref{lem:HOM(S)_before_HOM=HOM(S)}}
	\end{align*}
	\proofpart{\ref{item:weightedCSForMCFL:2} $⇒$ \ref{item:weightedCSForMCFL:1}}
	We use \cref{thm:multipleDyckIsMCFL,lem:weightSeparation}, and the closure of $\MCFG[k]$ under intersection with regular languages and application of homomorphisms.
\end{proof}

\begin{proposition}\label{cor:weightedCSbijection}
	For every $\mathcal{A}$-weighted MCFG $G$, there is a bijection between $D_G$ and $R(G_{\mathbb{B}}) ∩ \mathit{mD}(G_{\mathbb{B}})$.
\end{proposition}
\begin{proof}
	There are bijections
		between $D_G$ and $L(G_{\mathbb{B}})$ by claims in the proof of \cref{lem:weightSeparation},
		between $L(G_{\mathbb{B}})$ and $D_{G_{\mathbb{B}}}$ by claims in the proof of \cref{lem:weightSeparation}, and
		between $D_{G_{\mathbb{B}}}$ and $R(G_{\mathbb{B}}) ∩ \mathit{mD}(G_{\mathbb{B}})$ by \cref{cor:CSbijection}.
\end{proof}


\section{Conclusion and outlook}

We
	defined multiple Dyck languages using congruence relations in \cref{def:mDYCK},
	gave an algorithm to decide whether a word is in a given multiple Dyck language in \cref{alg:IsMultipleDyck}, and
	established that multiple Dyck languages with increasing maximal dimension form a hierarchy in \cref{prop:mDyckHierarchy}.

We obtained a weighted version of the CS representation of MCFLs for complete commutative strong bimonoids (\cref{thm:weightedCSForMCFL}) by separating the weights from the weighted MCFG and using Yoshinaka, Kaji, and Seki~\cite[Theorem~3]{YosKajSek10} for the unweighted part.

\Cref{thm:weightedCSForMCFL,cor:weightedCSbijection} may be used to derive a parsing algorithm for weighted multiple context-free grammars in the spirit of Hulden~\cite{Hul11}.


\bibliographystyle{alpha-doi}
\bibliography{cs-mcfg}

\newcommand{\etalchar}[1]{$^{#1}$}
\begin{thebibliography}{SMFK91}

\bibitem[CS63]{ChoSch63}
N.~Chomsky and M.~P. Schützenberger.
\newblock The algebraic theory of context-free languages.
\newblock {\em Computer Programming and Formal Systems, Studies in Logic},
  pages 118 -- 161, 1963.
\newblock DOI:
  \href{https://dx.doi.org/10.1016/S0049-237X(09)70104-1}{\nolinkurl{10.1016/S0049-237X(09)70104-1}}.

\bibitem[Den15]{Den15}
T.~Denkinger.
\newblock A {C}homsky-{S}chützenberger representation for weighted multiple
  context-free languages.
\newblock In {\em Proceedings of the 12th International Conference on
  Finite-State Methods and Natural Language Processing (FSMNLP 2015)}, 2015.
\newblock URL:
  \href{http://aclweb.org/anthology/W15-4803}{\nolinkurl{http://aclweb.org/anthology/W15-4803}}.

\bibitem[DPS79]{DusParSpe79}
J.~Duske, R.~Parchmann, and J.~Specht.
\newblock A homomorphic characterization of indexed languages.
\newblock {\em Elektronische Informationsverarbeitung und Kybernetik},
  15(4):187--195, 1979.

\bibitem[DSV10]{DroStueVog10}
M.~Droste, T.~Stüber, and H.~Vogler.
\newblock Weighted finite automata over strong bimonoids.
\newblock {\em Information Sciences}, 180(1):156--166, 2010.
\newblock DOI:
  \href{https://dx.doi.org/10.1016/j.ins.2009.09.003}{\nolinkurl{10.1016/j.ins.2009.09.003}}.

\bibitem[DV13]{DroVog13}
M.~Droste and H.~Vogler.
\newblock The {C}homsky-{S}chützenberger theorem for quantitative context-free
  languages.
\newblock In M.-P. Béal and O.~Carton, editors, {\em Proceedings of the 17th
  International Conference on Developments in Language Theory (DLT 2013)},
  volume 7907 of {\em Lecture Notes in Computer Science}, pages 203--214.
  Springer Berlin Heidelberg, 2013.
\newblock DOI:
  \href{https://dx.doi.org/10.1007/978-3-642-38771-5_19}{\nolinkurl{10.1007/978-3-642-38771-5_19}}.

\bibitem[FV15]{FraVou15}
S.~Fratani and E.~M. Voundy.
\newblock Context-free characterization of indexed languages.
\newblock {\em CoRR}, abs/1409.6112, 2015.
\newblock arXiv: \href{https://arxiv.org/abs/1409.6112}{\nolinkurl{1409.6112}}.

\bibitem[FV16]{FraVou16}
S.~Fratani and E.~M. Voundy.
\newblock Homomorphic characterizations of indexed languages.
\newblock In A.-H. Dediu, J.~Janoušek, C.~Martín-Vide, and B.~Truthe,
  editors, {\em Proceedings of the 10th International Conference on Language
  and Automata Theory and Applications (LATA 2015)}, pages 359--370. Springer
  Science + Business Media, 2016.
\newblock DOI:
  \href{https://dx.doi.org/10.1007/978-3-319-30000-9_28}{\nolinkurl{10.1007/978-3-319-30000-9_28}}.

\bibitem[HU69]{HopUll69}
J.~E. Hopcroft and J.~D. Ullman.
\newblock {\em Formal languages and their relation to automata}.
\newblock Addison-Wesley Longman Publishing Co., Inc., 1969.

\bibitem[HU79]{HopUll79}
J.~E. Hopcroft and J.~D. Ullman.
\newblock {\em Introduction to automata theory, languages and computation}.
\newblock Addison-Wesley, 1st edition, 1979.

\bibitem[Hul11]{Hul11}
M.~Hulden.
\newblock Parsing {CFG}s and {PCFG}s with a {C}homsky-{S}chützenberger
  representation.
\newblock In Z.~Vetulani, editor, {\em Human Language Technology. Challenges
  for Computer Science and Linguistics}, volume 6562 of {\em Lecture Notes in
  Computer Science}, pages 151--160. Springer Berlin Heidelberg, 2011.
\newblock DOI:
  \href{https://dx.doi.org/10.1007/978-3-642-20095-3_14}{\nolinkurl{10.1007/978-3-642-20095-3_14}}.

\bibitem[HV15]{HerVog15}
L.~Herrmann and H.~Vogler.
\newblock A {C}homsky-{S}chützenberger theorem for weighted automata with
  storage.
\newblock In A.~Maletti, editor, {\em Proceedings of the 6th International
  Conference on Algebraic Informatics (CAI 2015)}, volume 9270, pages 90--102.
  Springer International Publishing, 2015.
\newblock DOI:
  \href{https://dx.doi.org/10.1007/978-3-319-23021-4_11}{\nolinkurl{10.1007/978-3-319-23021-4_11}}.

\bibitem[Kal10]{Kal10}
L.~Kallmeyer.
\newblock {\em Parsing beyond context-free grammars}.
\newblock Springer, 2010.
\newblock DOI:
  \href{https://dx.doi.org/10.1007/978-3-642-14846-0}{\nolinkurl{10.1007/978-3-642-14846-0}}.

\bibitem[Kan14]{Kan14}
M.~Kanazawa.
\newblock Multidimensional trees and a {C}homsky-{S}chützenberger-weir
  representation theorem for simple context-free tree grammars.
\newblock {\em Journal of Logic and Computation}, Jun 2014.
\newblock DOI:
  \href{https://dx.doi.org/10.1093/logcom/exu043}{\nolinkurl{10.1093/logcom/exu043}}.

\bibitem[Mic01a]{Mic01a}
J.~Michaelis.
\newblock Derivational minimalism is mildly context-sensitive.
\newblock In M.~Moortgat, editor, {\em Logical Aspects of Computational
  Linguistics}, volume 2014 of {\em Lecture Notes in Computer Science}, pages
  179--198. Springer Berlin Heidelberg, 2001.
\newblock DOI:
  \href{https://dx.doi.org/10.1007/3-540-45738-0_11}{\nolinkurl{10.1007/3-540-45738-0_11}}.

\bibitem[Mic01b]{Mic01}
J.~Michaelis.
\newblock Transforming linear context-free rewriting systems into minimalist
  grammars.
\newblock In P.~Groote, G.~Morrill, and C.~Retoré, editors, {\em Logical
  Aspects of Computational Linguistics}, volume 2099 of {\em Lecture Notes in
  Computer Science}, pages 228--244. Springer Berlin Heidelberg, 2001.
\newblock DOI:
  \href{https://dx.doi.org/10.1007/3-540-48199-0_14}{\nolinkurl{10.1007/3-540-48199-0_14}}.

\bibitem[Moh00]{Moh00}
M.~Mohri.
\newblock Minimization algorithms for sequential transducers.
\newblock {\em Theoretical Computer Science}, 234(1--2):177--201, Mar 2000.
\newblock DOI:
  \href{https://dx.doi.org/10.1016/s0304-3975(98)00115-7}{\nolinkurl{10.1016/s0304-3975(98)00115-7}}.

\bibitem[Sal73]{sal73}
A.~Salomaa.
\newblock {\em Formal languages}.
\newblock Academic Press, 1973.

\bibitem[SMFK91]{SekMatFujKas91}
H.~Seki, T.~Matsumura, M.~Fujii, and T.~Kasami.
\newblock On multiple context-free grammars.
\newblock {\em Theoretical Computer Science}, 88(2):191--229, 1991.
\newblock DOI:
  \href{https://dx.doi.org/10.1016/0304-3975(91)90374-B}{\nolinkurl{10.1016/0304-3975(91)90374-B}}.

\bibitem[SNK{\etalchar{+}}93]{SekNakKajAndKas93}
H.~Seki, R.~Nakanishi, Y.~Kaji, S.~Ando, and T.~Kasami.
\newblock Parallel multiple context-free grammars, finite-state translation
  systems, and polynomial-time recognizable subclasses of lexical-functional
  grammars.
\newblock In {\em Proceedings of the 31st Annual Meeting of the Association for
  Computational Linguistics (ACL 1993)}, ACL '93, pages 130--139, Stroudsburg,
  PA, USA, 1993. Association for Computational Linguistics.
\newblock DOI:
  \href{https://dx.doi.org/10.3115/981574.981592}{\nolinkurl{10.3115/981574.981592}}.

\bibitem[SS78]{SalSoi78}
A.~Salomaa and M.~Soittola.
\newblock {\em Automata-Theoretic Aspects of Formal Power Series}.
\newblock Springer New York, 1978.
\newblock DOI:
  \href{https://dx.doi.org/10.1007/978-1-4612-6264-0}{\nolinkurl{10.1007/978-1-4612-6264-0}}.

\bibitem[VS87]{Vij87}
K.~Vijay-Shanker.
\newblock {\em A study of tree adjoining grammars}.
\newblock PhD thesis, 1987.

\bibitem[VSWJ86]{VijWeiJos86}
K.~Vijay-Shanker, D.~J. Weir, and A.~K. Joshi.
\newblock Tree adjoining and head wrapping.
\newblock In {\em Proceedings of the 11th Conference on Computational
  Linguistics (COLING 1986)}, pages 202--207. Association for Computational
  Linguistics, 1986.
\newblock DOI:
  \href{https://dx.doi.org/10.3115/991365.991425}{\nolinkurl{10.3115/991365.991425}}.

\bibitem[Wei88]{Wei88}
D.~J. Weir.
\newblock {\em Characterizing mildly context-sensitive grammar formalisms}.
\newblock PhD thesis, 1988.

\bibitem[WJ88]{WeiJos88}
D.~J. Weir and A.~K. Joshi.
\newblock Combinatory categorial grammars: Generative power and relationship to
  linear context-free rewriting systems.
\newblock In {\em Proceedings of the 26th Annual Meeting of the Association for
  Computational Linguistics (ACL 1988)}, pages 278--285. Association for
  Computational Linguistics, 1988.
\newblock DOI:
  \href{https://dx.doi.org/10.3115/982023.982057}{\nolinkurl{10.3115/982023.982057}}.

\bibitem[YKS10]{YosKajSek10}
R.~Yoshinaka, Y.~Kaji, and H.~Seki.
\newblock {C}homsky-{S}chützenberger-type characterization of multiple
  context-free languages.
\newblock In A.-H. Dediu, H.~Fernau, and C.~Martín-Vide, editors, {\em
  Language and Automata Theory and Applications}, pages 596--607. Springer,
  2010.
\newblock DOI:
  \href{https://dx.doi.org/10.1007/978-3-642-13089-2_50}{\nolinkurl{10.1007/978-3-642-13089-2_50}}.

\end{thebibliography}

\end{document}